%% file: Generalization-TSP-revised.tex
\documentclass[journal]{IEEEtran}
\IEEEoverridecommandlockouts
 \usepackage{etex}
\usepackage[utf8]{inputenc}
\usepackage[pdftex]{graphicx}
\usepackage{graphicx}
\usepackage[T1]{fontenc}    
\usepackage{hyperref}       
\usepackage{url}            
\usepackage{booktabs}       
\usepackage{amsfonts}       
\usepackage{nicefrac}       
\usepackage{microtype}      
\usepackage{graphics, times, cite}
\usepackage{epstopdf}
\usepackage{tikz}
\usepackage{setspace}

\usepackage{bbm}
\usepackage{mathrsfs}
\usepackage{adjustbox}
\usepackage{makecell}
\usepackage{wrapfig}

\usepackage{amssymb,amsmath,amsfonts,bm,epsfig,epstopdf,amsthm}
\usepackage{rotating,setspace,latexsym,epsf,color,dsfont,caption,mathtools}
\usepackage{caption}
\usepackage{subcaption} 
\usepackage{algorithm}
\usepackage{algorithmic}
\usepackage[shortlabels]{enumitem}
\input{latex_resources/mysymbol.sty}
\input{latex_resources/my_sections}

\newif\ifcoloron
\colorontrue 

\usepackage{caption}

\newcommand{\QED}{\hfill\ensuremath{\blacksquare}}

\def\E{\E}

\newcommand{\EmpRiskGraph}{R_E} 
\newcommand{\StatRiskGraph}{R_S} 

\newcommand{\PEmpRiskGraph}{P_E^{*}} 
\newcommand{\PStatRiskGraph}{P_S^{*}} 

\newcommand{\GeneralizationGap}{GA} 

{\title{{Generalization of Geometric Graph Neural Networks with Lipschitz Loss Functions}}}
\author{Zhiyang Wang \quad Juan Cervi\~no \quad Alejandro Ribeiro \thanks{Supported by NSF TILOS and THEORINET Simons. Zhiyang Wang and Alejandro Ribeiro are with the Department of Electrical
and Systems Engineering, University of Pennsylvania, Philadelphia, PA 19104
USA (email: zhiyangw@seas.upenn.edu; aribeiro@seas.upenn.edu).
Juan Cerviño was with the Department of Electrical and Systems Engineering, University of Pennsylvania, Philadelphia, PA 19104 USA. He is
now with the Laboratory for Information and Decision Systems (LIDS),
Massachusetts Institute of Technology, Cambridge, MA 02139 USA (email:
jcervino@mit.edu). Preliminary results are presented in \cite{wangcervino2024}. }
}

\begin{document}

\maketitle
\begin{abstract}
In this paper, we study the generalization capabilities of geometric graph neural networks (GNNs). We consider GNNs over a geometric graph constructed from a finite set of randomly sampled points over an embedded manifold with topological information captured. We prove a generalization gap between the optimal empirical risk and the optimal statistical risk of this GNN, which decreases with the number of sampled points from the manifold and increases with the dimension of the underlying manifold. This generalization gap ensures that the GNN trained on a graph on a set of sampled points can be utilized to process other unseen graphs constructed from the same underlying manifold. The most important observation is that the generalization capability can be realized with one large graph instead of being limited to the size of the graph as in previous results. The generalization gap is derived based on the non-asymptotic convergence result of a GNN on the sampled graph to the underlying manifold neural networks (MNNs). We verify this theoretical result with experiments on multiple real-world datasets.
\end{abstract}


\section{Introduction}
\label{sec:label}
\input{introduction}

\section{Preliminaries}
\label{sec:preliminary}
\input{preliminary}
\section{Convergence and Generalization Analysis of GNNs via MNNs}
\label{sec:generalization}
\input{generalization}

\section{Simulations}
\label{sec:sim}
\input{simulations}

\section{Conclusion}
\label{sec:conclusion}
\input{conclusion}

\urlstyle{same}
\bibliographystyle{IEEEtran}
\bibliography{bib}
\newpage
\appendix
 {\section{Appendix}
 \input{appendix}}
 
\end{document}

%% file: latex_resources/my_sections.tex
\usepackage{needspace}

\newcommand{\myparagraph}[1]{\needspace{1\baselineskip}\medskip\noindent {\bf #1.}}



%% file: introduction.tex

Graph structures can provide models for many modern datasets as they can capture the internal relationships among the data. Convolutional filters and neural networks on graphs \cite{sandryhaila2013discrete, gama2019convolutional,ortega2018graph} have become the top tool to process signals over graphs. A graph as a discrete model can naturally represent a discrete data structure. Examples include but not limited to social networks \cite{wu2022graph}, protein structures \cite{yin2023coco}, multi-agent control \cite{gosrich2022coverage}. In many practical scenarios, graphs can be seen as samples from a manifold, which can approximate the continuous topological space with finite samples, as in the case of point clouds \cite{bronstein2017geometric}, data manifolds \cite{sharma2022scaling}, and irregular space navigation \cite{CERVINO_ICML}. In this context, the graph convolutional filters and GNNs on these sampled graphs can be proved to converge to the manifold convolutional filters and manifold neural networks \cite{wang2022convolutional, wang2023geometric}. This implies that convolutional structures on graphs with finite sampled points can approximate the counterparts on the underlying manifold.

The main technical contribution of this paper is analyzing the generalization capabilities of GNNs operated on a graph sampled from a manifold on the node level. {We show that GNNs trained on a graph can be implemented on other unseen graphs sampled from the same manifold.} We consider this manifold as a common underlying structure for the graphs and derive a 
generalization bound that decreases with the number of nodes. This common manifold limit model lifts the restriction to the graph size in the generalization analyses in previous works.

\myparagraph{GNN Generalization via Manifold} To present our main result formally, a set of $N$ i.i.d. points $X_N$ is given over an embedded manifold $\ccalM$. Input and target graph signals $\bbx_N$ and $\bby_N$ are sampled on $X_N$ from underlying manifold signals. The goal is to learn a GNN $\bm\Phi$ that estimates $\bby_N$ with $\bm\Phi(\bbH, \bbx_N)$ where $\bbH\in\ccalH$ represents the filter parameter set. We use a $L_2$ loss function $\ell$ to measure the estimation performance between true target $\bby_N$ and the estimated target $\bm\Phi(\bbH, \bbx_N)$. Practically, the GNN is trained to minimize an empirical risk written as 
\begin{equation}\label{eqn:empirical_risk}
  \EmpRiskGraph (\bbH) = \ell(\bm\Phi(\bbH,\bbx_{N}),\bby_{N}).
\end{equation}
While theoretically, the GNN aims to minimize a statistical risk, which is written as 
\begin{equation}\label{eqn:statistical_risk}
  \StatRiskGraph (\bbH)  =  \mathbb{E}_{X_{N}}\left[ \ell\left( \bm\Phi(\bbH,\bbx_{N}), \bby_{N}\right)\right]. 
\end{equation}
The generalization gap is defined to be 
\begin{equation}\label{eqn:generalization_gap}
  GA= \min_{\bbH\in \ccalH}  \StatRiskGraph (\bbH)-  \min_{\bbH\in\ccalH} \EmpRiskGraph (\bbH) .
\end{equation}
We analyze the generalization gap in \eqref{eqn:generalization_gap} through the convergence of GNNs to the neural networks built on the underlying manifold, which is manifold neural network \cite{wang2022convolutional, wang2024stability}. With the non-asymptotic convergence results, we can prove that the generalization gap between the GNN performances on the finite $N$ training points of the sampled graph and the true distribution of these $N$ points is small and decreases with $N$. 
\begin{theorem*}[Informal]
    \label{thm:inform-generalization}
    Consider a graph constructed on $N$ i.i.d. randomly sampled points over a $d$-dimensional manifold $\ccalM$ with respect to the measure $\mu$ over the manifold. Then, the generalization gap of a GNN trained on this graph satisfies with probability $1-\delta$ that
    \begin{equation}
    \label{eqn:informal}
        GA =\ccalO\left( \left( \frac{\log N/\delta}{N}\right)^{\frac{1}{d+4}}\right).
    \end{equation}
\end{theorem*}
The equation \eqref{eqn:informal} shows that the generalization gap decreases approximately polynomially with respect to the number of points $N$ while the exponent is related to the dimension of the underlying manifold. We can also observe that the generalization gap increases with the dimension of the manifold, i.e. the complexity of the underlying continuous model. This indicates that for a high dimensional manifold, we need more number of sampled points to guarantee the generalization capability.

\myparagraph{Related Works and Significance} Analyses have been carried out in \cite{wang2023geometric, ruiz2020graphon, ruiz2023transferability, levie2021transferability, maskey2023transferability}, where transferability of GNNs are analyzed by comparing the output difference of GNNs on different sizes of graphs when graphs converge to a limit model as manifold or graphon without generalization analysis. In \cite{CERVINO_ICASSP, CERVINO_TSP}, the authors show how increasing the size of the graph as the GNN learns, generalizes to the large-scale graph. 

In \cite{scarselli2018vapnik}, the authors prove the generalization bound of GNNs with VC-dimension which is commonly used for convolutional neural networks. The authors in \cite{verma2019stability} analyze the generalization of a single-layer GNN based on stability analysis, which is further extended to a multi-layer GNN in \cite{zhou2021generalization}. {The authors in \cite{yang2024deeper} extend this analysis to deeper GNNs and analyze the approximation theory of GNNs in \cite{yang2024bridging}.}
In \cite{ma2021subgroup}, the authors give a novel PAC-Bayesian analysis on the generalization bound of GNNs across arbitrary subgroups of training and testing datasets.
The authors derive generalization bounds for GNNs via transductive uniform stability and transductive Rademacher complexity in \cite{esser2021learning, cong2021provable, tang2023towards}. The authors in \cite{yehudai2021local} propose a size generalization analysis of GNNs correlated to the discrepancy between local distributions of graphs. The authors in \cite{shi2024homophily} analyze the transductive generalization of a simple GNN with a contextual stochastic block model. {The results of these works ignore the underlying graph structures, which leads to generalization bounds increasing with the number of nodes -- a behavior contrary to the one experienced in practice.} We consider a continuous manifold model as a generative model for the graphs, which provides powerful theoretical support and is realistic to model real-world applications.

Generalization analysis has also been carried out on the graph-level when considering the graph classification problem.
In \cite{liao2020pac}, the authors provide a generalization bound based on PAC-Bayesian analysis with the bound depending on the maximum degree of the graph and the spectral norms. In \cite{garg2020generalization}, the authors provide a generalization bound on message passing networks comparable to Rademacher bounds in recurrent neural networks. The previous generalization gaps again scale with the size of graphs without capturing the underlying common graph structures.  A generalization analysis is carried out in \cite{maskey2022generalization, maskey2024generalization}, where the graphs are on message-passing neural networks on graphs randomly sampled from a collection of template random graph models. We study the limit of graphs as a manifold, which is {widely accepted in physical dynamics \cite{talmon2015manifold}, and images (the so-called image manifold) \cite{peyre2009manifold, osher2017low}} to process high-dimensional data. In \cite{wangcervnino2024neurips}, the authors also study the generalization analysis of GNNs from a manifold perspective both on node and graph levels. {This work studies the generalization of GNNs to the MNNs directly in a statistical way, while our work focuses on the generalization capability when we only have access to sampled graphs from the manifold. This provides a more direct and applicable framework for evaluating practical performances as we cannot access the continuous manifold model directly.} Most importantly, only one graph is needed to train the GNN and this allows the generalization of all the other unseen sampled graphs from the manifold.

\myparagraph{Organization} Section \ref{sec:preliminary} introduces preliminary concepts of graph convolutions and graph neural networks (GNNs) as well as manifold convolutions and manifold neural networks (MNNs). Section \ref{sec:generalization} introduces the construction of an $\epsilon$-graph where only nodes that are close enough are connected based on sampled points from the manifold. We first establish the convergence result of GNNs constructed on this $\epsilon$-graph to the underlying MNNs with low-pass filters and normalized Lipschitz nonlinearities. 
Based on the convergence result, we present the generalization gap of GNNs on the constructed $\epsilon$-graphs which attest that GNN trained on one $\epsilon$-graph can generalize to other unseen $\epsilon$-graphs from the same manifold. Section \ref{sec:sim} illustrates the generalization results of Section \ref{sec:generalization} with verifications on Arxiv and Cora datasets. Section \ref{sec:conclusion} concludes the paper. Proofs are deferred to the appendices.

%% file: preliminary.tex

Let us start with the basic definitions of graph neural networks and manifold neural networks. 

\subsection{Graph Convolutions and Graph Neural Networks}

An undirected graph $\bbG = (\ccalV, \ccalE, \ccalW)$ contains $N$ nodes with a node set $\ccalV$ and edge set $\ccalE \subseteq \ccalV\times \ccalV$. The weights of the edges are assigned by $\ccalW: \ccalE \rightarrow \reals$. Graph signals $\bbx \in \reals^N$ map values to each node. A graph shift operator (GSO) \cite{ortega2018graph} $\bbS \in \reals^{N\times N}$ is a graph matrix with $[\bbS]_{ij} \neq 0$ if and only if $(i,j)\in \ccalE$ or $i=j$, e.g., the graph Laplacian $\bbL$. The GSO can shift or diffuse signals to each node by aggregating signal values of neighbors. A graph convolution is defined based on a consecutive graph shift operation. A graph convolutional filter $\bbh_\bbG$ \cite{gama2019convolutional,ortega2018graph,sandryhaila2013discrete} with filter coefficients $\{h_k\}_{k=0}^{K-1}$ is formally defined as  
\begin{equation}
    \label{eqn:graph_convolution}
\bbh_\bbG(\bbS) \bbx = \sum_{k=0}^{K-1} h_k \bbS^k \bbx .
\end{equation}
We replace $\bbS$ with the spectral decomposition $\bbS = \bbV \bm\Lambda \bbV^H$, where $\bbV$ is the eigenvector matrix and $\bm\Lambda$ is a diagonal matrix with eigenvalues as the entries. We observe that the spectral representation of a graph filter is
\begin{equation}
    \label{eqn:graph_convolution_spectral}
    \bbV^H \bbh_\bbG(\bbS) \bbx =  \sum_{k=1}^{K-1} h_k \bm\Lambda^k \bbV^H \bbx = h(\bm\Lambda)\bbV^H \bbx.
\end{equation}
This leads to a point-wise frequency response of the graph convolution, which is $h(\lambda)= \sum_{k=0}^{K-1} h_k \lambda^k$, depending only on the weights $\{h_k\}_{k=0}^{K-1}$ and on the eigenvalues $\bm\Lambda$ of $\bbS$.

A graph neural network (GNN) is composed of cascading layers that each consists of a bank of graph convolutional filters followed by a point-wise nonlinearity $\sigma:\reals\to\reals$. Specifically, the $l$-th layer of a GNN that produces $F_l$ output features $\{\bbx_l^p\}_{p=1}^{F_l}$ with $F_{l-1}$ input features $\{\bbx^q_{l-1}\}_{q=1}^{F_{l-1}}$ is written as 
\begin{equation} \label{eqn:gnn}
    \bbx_l^p = \sigma\left(\sum_{q=1}^{F_{l-1}} \bbh_\bbG^{lpq}(\bbS) \bbx^q_{l-1} \right),
\end{equation}
for each layer $l=1,2\cdots, L$, {with $\bbx_l^p\in \reals^N$}. The graph filter $\bbh_\bbG^{lpq}(\bbS)$ maps the $q$-th feature of layer $l-1$ to the $p$-th feature of layer $l$ as \eqref{eqn:graph_convolution}. We denote the GNN like \eqref{eqn:gnn} as a mapping $\bm\Phi(\bbH, \bbS, \bbx)$ for the ease of presentation, where $\bbH\in \ccalH\subset \reals^P$ denotes a set of the graph filter coefficients at all layers and $\ccalH$ denotes the set of all possible parameter sets.

\subsection{Manifold Convolutions and Manifold Neural Networks}

We consider a $d$-dimensional compact, smooth and differentiable submanifold $\ccalM$ embedded in {$\reals^\mathsf{D}$}. The embedding induces a probability measure $\mu$ over the manifold with density function $\rho:\ccalM\rightarrow(0,\infty)$, which is assumed to be bounded as $0<\rho_{min}\leq \rho(x) \leq \rho_{max}<\infty$ for each $x\in\ccalM$. Manifold signals \cite{wang2024stability} are likewise defined as scalar functions $f:\ccalM\rightarrow \reals$. We use $L^2(\ccalM)$ to denote $L^2$ functions over $\ccalM$ with respect to measure $\mu$. The inner product of signals $f, g\in L^2(\ccalM)$ is defined as 
\begin{equation}\label{eqn:innerproduct}
    \langle f,g \rangle_{\ccalM}=\int_\ccalM f(x)g(x) \text{d}\mu(x), 
\end{equation}
with the $L^2$ norm defined as $\|f\|^2_\ccalM = \langle f, f\rangle_\ccalM$. The manifold with density $\rho$ is endowed with a weighted Laplace operator \cite{grigor2006heat}, which generalizes the Laplace-Beltrami operator as
\begin{equation}
    \label{eqn:weight-Laplace}
    \ccalL_\rho f = -\frac{1}{2\rho} \text{div}(\rho^2 \nabla f),
\end{equation}
with $\text{div}$ being the divergence operator of $\ccalM$ and $\nabla$ being the gradient operator of $\ccalM$ \cite{bronstein2017geometric}. Manifold shift operation is defined relying on the Laplace operator $\ccalL_\rho$ and on the heat diffusion process over the manifold (see \cite{wang2024stability} for a detailed exposition). For a manifold signal $f\in L^2(\ccalM)$, the manifold shift can be explicitly written as $e^{-\ccalL_\rho} f$. Analogous to graph convolution, manifold convolution \cite{wang2024stability} can be defined in a shift-and-sum manner as  
\begin{align} \label{eqn:manifold-convolution}
   g(x) = \bbh(\ccalL_\rho)f(x) =\sum_{k=0}^{K-1} h_ke^{-k\ccalL_\rho}f(x) \text{.}
\end{align}
Consider the Laplace operator is self-adjoint and positive-semidefinite and the manifold $\ccalM$ is compact, $\ccalL_\rho$ has real, positive and discrete eigenvalues $\{\lambda_i\}_{i=1}^\infty$, written as $\ccalL_\rho \bm\phi_i =\lambda_i \bm\phi_i$ where $\bm\phi_i$ is the eigenfunction associated with eigenvalue $\lambda_i$, which can be seen as the osillation mode. The eigenvalues are ordered in increasing order as {$0\leq \lambda_1\leq \lambda_2\leq \lambda_3\leq \hdots$}, and the eigenfunctions are orthonormal and form an eigenbasis of $L^2(\ccalM)$. When mapping a manifold signal onto one of the eigenbasis $\bm\phi_i$, we have  
\begin{equation}
 [\hat{f} ]_i=\langle f, \bm\phi_i\rangle_{L^2(\ccalM)}=\int_{\ccalM} f(x)\bm\phi_i(x) \text{d}\mu(x),
\end{equation}
the manifold convolution can be written in the spectral domain point-wisely as
\begin{align}
     [\hat{g} ]_i = \sum_{k=0}^{K-1} h_k e^{-k\lambda_i}   [\hat{f} ]_i\text{.}
\end{align}
Hence, the frequency response of manifold filter is given by $\hat{h}(\lambda)=\sum_{k=0}^{K-1} h_k e^{-k\lambda}$, depending only on the filter coefficients $h_k$ and eigenvalues of $\ccalL_\rho$ when $\lambda=\lambda_i$. This can be seen as a generalization of graph convolution when we take the Taylor expansion of $e^{-k}$. To get the representation of the manifold filter in the spectral domain, we can sum over all frequency components and project back to the spatial domain, which is 
\begin{equation}\label{eqn:spectral-filter}
    g =\bbh(\ccalL)f= \sum_{i=1}^{\infty} \hat{h}(\lambda_i) [\hat{f} ]_i \bm\phi_i .
\end{equation}
Manifold neural networks (MNNs) are built by cascading $L$ layers, each of which consists of a bank of manifold filters and a pointwise nonlinearity $\sigma$. Each layer $l=1,2\cdots, L$ can be explicitly denoted as
\begin{equation}\label{eqn:mnn}
f_l^p(x) = \sigma\left( \sum_{q=1}^{F_{l-1}} \bbh_l^{pq}(\ccalL_\rho) f_{l-1}^q(x)\right),
\end{equation}
where $f_{l-1}^q$ is an input feature and $f_l^p$, $1 \leq p \leq F_l$ is a an output feature. In each layer manifold filters maps $F_{l-1}$ input features to $F_l$ output features. To represent the MNN succinctly, we {group all learnable parametes, and we} denote the mapping as $\bbPhi(\bbH,\ccalL_\rho,f)$, where $\bbH\in\ccalH\subset\reals^P$  is a filter parameter set of the manifold filters.

%% file: generalization.tex
Suppose we are given an embedded manifold $\ccalM\subset \reals^\mathsf{N}$ with input manifold signal $f\in L^2(\ccalM)$ and target manifold signal $g\in L^2(\ccalM)$ attached to it. An MNN, as defined in \eqref{eqn:mnn}, predicts the target signal $g$ with $\bm\Phi(\bbH, \ccalL_\rho, f)$ where $\bbH\in\ccalH\subset\reals^P$ is the set of filter coefficients, $\ccalL_\rho$ is the weighted Laplacian defined in \eqref{eqn:weight-Laplace}. A positive loss function is denoted as $\ell(\bm\Phi(\bbH, \ccalL_\rho, f), g)$ to measure the estimation performance.

Suppose we are given a pair of graph signal and graph $(\bbx_N, \bbG_N)$ along with target output graph signal $\bby_N \in \reals^N$. Graph $\bbG_N$ is constructed based on a set of $N$ i.i.d. randomly sampled points $X_N=\{x_N^1, x_N^2,\cdots, x_N^{N}\}$ according to measure $\mu$ over the underlying manifold $\ccalM$. These $N$ sampled points are seen as nodes. Further, every pair of nodes $(x_N^i, x_N^j)$ is connected by an edge with weight value $[\bbW_N]_{ij}, \bbW_N\in \reals^{N \times N }$ determined by a function $K_\epsilon$ of their Euclidean distance $\|x_N^i - x_N^j\|$ \cite{calder2022improved}, which is explicitly written as 
\begin{equation}
\label{eqn:edge-weight}
[\bbW_N]_{ij} = K_\epsilon(x_N^i, x_N^j) = \frac{\alpha_d}{(d+2)N \epsilon^{d+2} } \mathbbm{1}_{[0,1]}\left( \frac{\|x_N^i - x_N^j\|}{\epsilon} \right),
\end{equation}
where $\alpha_d$ is the volume of the $d$-dimensional Euclidean unit ball and $\mathbbm{1}$ stands for a characteristic function, {i.e. the function value is $1$ if and only if $\frac{\|x_N^i - x_N^j\|}{\epsilon} $ falls in $[0,1]$}. {We call graph $\bbG_N$ as a geometric graph as it contains the geometric information embedded in the underlying manifold model.} Graph input and target output signals $\bbx_N, \bby_N\in\reals^{N }$ are supported on sampled points $X_N$ belonging to $L^2(X_N)$, whose values are sampled from manifold input signal $f$ and target signal $g$ respectively, written explicitly as
\begin{align}
\label{eqn:sample}
    [\bbx_N]_i = f(x_N^i),\quad [\bby_N]_i = g(x_N^i).
\end{align}
We define a sampling operator $\bbP_N: L^2(\ccalM)\rightarrow L^2(X_N)$ to represent the mappings in \eqref{eqn:sample}, i.e. $\bbx_N = \bbP_N f$ and $\bby_N =\bbP_N g$.
As introduced in Sec. \ref{sec:preliminary}, we denote a GNN mapping as $\bm\Phi(\bbH,\bbL_N,\bbx_N)$, where $\bbH\in\ccalH\subset\reals^P$ is the set of filter coefficients, $\bbL_N = \text{diag}(\bbW_N \mathbf{1}) - \bbW_N$ as the graph Laplacian which is implemented as a GSO, $\bbx_N$ as the input graph signal. 
\subsection{Convergence of GNNs to MNNs}
We first show that the difference between the outputs of MNN $\bm\Phi(\bbH,\ccalL_\rho, f)$ and the GNN on the sampled graph with $N$ nodes  $\bm\Phi(\bbH,\bbL_N,\bbx_N)$ can be bounded. The bound decreases with the number of nodes sampled from the manifold. 

Suppose the input manifold signal $f$ is $\lambda_M$-bandlimited, which is explicitly defined as follows.
\begin{definition}
\label{def:band}
      A manifold signal $f$ is $\lambda_M$-bandlimited if for all eigenpairs $\{\lambda_i, \bm\phi_i\}_{i=1}^\infty$ of the weighted Laplacian $\ccalL_\rho$ when $\lambda_i>\lambda_M$, we have $\langle f, \bm\phi_i \rangle_\ccalM = 0$.
\end{definition}
We denote $M$ as the cardinality of the limited spectrum of $\ccalL_\rho$, i.e. $M =\# \{\lambda_i <\lambda_M\}$. This indicates that $f$ consists of finite frequency components over the manifold without components in the high oscillation modes. We also need to put an assumption on the frequency response function of the filters as follows.
\begin{definition}
\label{def:filter}
A filter is a low pass filter if its frequency response satisfies 
\begin{equation}
        \label{eqn:low-pass}
        \left|\hat{h}(a)\right| =\ccalO\left(a^{-d}\right),\quad a\in(0,\infty).
    \end{equation}
\end{definition}
The low pass filter assumption \eqref{eqn:low-pass} is reasonable given that most real-world graph processes are low frequency \cite{ramakrishna2020user}.
For the nonlinearity functions utilized in the GNNs, we need to make the following assumption.
\begin{assumption}(Normalized Lipschitz nonlinearity)\label{ass:activation}
 The nonlinearity $\sigma$ is normalized Lipschitz continuous, i.e., $|\sigma(a)-\sigma(b)|\leq |a-b|$, with $\sigma(0)=0$.
\end{assumption}
We note that this assumption is reasonable considering most common activation functions used in practice are normalized Lipschitz, such as ReLU, modulus and sigmoid.


With the low-pass filters and normalized Lipschitz non-linearities in both the GNN and MNN, we are ready to prove a difference bound between the outputs of a GNN operating on a sampled graph over $\ccalM$ and the outputs of an MNN on $\ccalM$.
\begin{proposition}
\label{prop:prob-diff}
    Let $\ccalM\subset \reals^\mathsf{N}$ be an embedded manifold with weighted Laplace operator $\ccalL_\rho$ and a $\lambda_M$-bandlimited manifold signal $f$. Consider a pair of graph and graph signal $(\bbx_N, \bbG_N)$ with $N$ nodes sampled i.i.d. with measure $\mu$ over $\ccalM$. The graph Laplacian $\bbL_N$ is calculated with \eqref{eqn:edge-weight}. Let $\bm\Phi(\bbH,\ccalL_\rho,\cdot)$ be a single layer MNN on $\ccalM$ \eqref{eqn:mnn} with single input and output features. Let $\bm\Phi(\bbH,\bbL_N,\cdot)$ be the GNN with the same architecture applied to the graph $\bbG_N$. Then, with the filters as low-pass and nonlinearities as normalized Lipschitz continuous, it holds in probability at least $1-\delta$ that 
    \begin{align}
         \label{eqn:prob-diff}
    &\nonumber\|\bm\Phi(\bbH,\bbL_N, \bbP_N f) - \bbP_N\bm\Phi(\bbH, \ccalL_\rho ,f ) \|_2
\leq \\ \nonumber
&    C_1 \left(\frac{\log\frac{C_1N}{\delta}}{N} \right)^{\frac{1}{d+4}} +  C_2\left(\frac{\log\frac{C_1N}{\delta}}{N} \right)^{\frac{1}{d+4}} \theta_M^{-1} \\
    &\qquad \qquad \qquad \qquad\qquad   + C_3 \sqrt{\frac{\log(1/\delta)}{N}} + C_4 M^{-1},
    \end{align}
    where $C_1$, $C_2$, $C_3$ and $C_4$ are constants defined in Appendix \ref{app:prop1} and $\theta_M = \min_{i=1,2\cdots, M} |\lambda_i -\lambda_{i+1}|$.
\end{proposition}
\begin{corollary}
\label{cor:converge}
    The difference of the GNN $\bm\Phi(\bbH,\bbL_N,\cdot)$ on a graph $\bbG_N$ sampled from the manifold $\ccalM$ and MNN $\bm\Phi(\bbH,\ccalL_\rho,\cdot)$ converges to zero as $N$ goes to infinity {when the input signal has large enough bandwidth}.
\end{corollary}
\begin{proof}
    We denote the four terms in \eqref{eqn:prob-diff} as $A_1(N)$, $A_2(M,N)$, $A_3(N)$ and $A_4(M)$. For every $\delta>0$, we can choose some $M_0$ such at $A_4(M_0)<\delta /2$. There exists some $n_0$ such that for all $N>n_0$, $A_1(N)+A_2(M_0,N)+A_3(N)<\delta/2$. Therefore, this satisfies the definition of the convergence, which implies that for every $\delta>0$, there exists some $n_0$ so that for all $N>n_0$, we have $A_1(N)+ A_2(M_0,N)+A_3(N)+A_4(M_0)<\delta$. 
\end{proof}
Proposition \ref{prop:prob-diff} shows that the output difference of GNN on the sampled graph and the underlying MNN is bounded in high probability with the bound decreasing with the number of sampled points $N$, i.e. the number of nodes in the graph. According to the first and the third terms, the generalization gap decreases when the number of nodes $N$ increases and the dimension of the manifold $d$ decreases. This is reasonable as low-dimensional manifold is a model with low complexity and more sampled points from the manifold provide a better approximation. The second and the fourth terms also reflect the relationship with the bandwidth of the input manifold signal, with $\theta_M^{-1}$ increasing with $M$ according to Weyl's law \cite{arendt2009weyl} and $M^{-1}$ decreasing with $M$. In the second term, though $\theta_M^{-1}$ increases with $M$, a large enough $N$ can still make this term controllable. As Corollary \ref{cor:converge} shows, even as the spectrum becomes larger, there always exists some $N$ large enough to make the difference go to zero as $N$ increases. This attests to the convergence of GNN on the sampled graph to MNN on the underlying manifold. As this difference bound holds in high probability, we can naturally derive the difference bound in expectation of $N$ randomly sampled points as follows.
\begin{corollary}
\label{cor:expect-diff}
    The difference bound between GNN and MNN also holds in expectation since each node in $X_N$ is sampled i.i.d. according to measure $\mu$ over $\ccalM$
    \begin{align}
    \label{eqn:expect-diff}
        &  \mathop{\mathbb{E}}_{X_N \sim \mu^N} \left[\|\bm\Phi(\bbH,\bbL_N, \bbP_N f) - \bbP_N\bm\Phi(\bbH, \ccalL_\rho ,f ) \|_2\right] \leq \\\nonumber  &C' N^{-\frac{1}{d+4}} + C''N^{-\frac{1}{2}} +C'''\left(\frac{\log N}{N}\right)^{\frac{1}{d+4}} + \bar M e^{-N/C}\sqrt{N},
    \end{align}
    where $C'$, $C''$, $C'''$ and $\bar M$ are specified in Appendix \ref{app:cor2}.
\end{corollary}
We can observe that the expected output difference decreases with $N$ -- the number of nodes sampled from the underlying manifold, while increasing with $d$ -- the manifold dimension, i.e. the model complexity.  Equipped with these convergence results both in high probability and in expectation, we are ready to derive the bound of generalization gap of GNNs.

\subsection{Generalization of GNNs} 
Suppose a GNN is trained over the given graph $\bbG_N$ and graph signal $\bbx_N$. { A positive loss function $\ell$ is used to measure the estimation performance.  We assume that the loss function is normalized Lipschitz continuous defined as follows}
{
\begin{assumption}
    (Normalized Lipschitz loss function)\label{ass:loss}
 The loss function $\ell$ is normalized Lipschitz continuous, i.e., $|\ell( \bby_i, \bby)-\ell( \bby_j, \bby)|\leq \|\bby_i- \bby_j\|_2$, with $\ell(\bby, \bby)=0$.
\end{assumption}
}
The training of GNNs intends to minimize the empirical risk $\PEmpRiskGraph =\min\limits_{\bbH\in \ccalH}  \EmpRiskGraph (\bbH)$ with $\EmpRiskGraph (\bbH)$ defined as 
{\begin{align}\label{P:empirical_risk_graph}
    \EmpRiskGraph (\bbH)  := \ell(\bm\Phi(\bbH, \bbL_{N},\bbx_{N}),\bby_{N}).
\end{align}}

Given the fact that we only have access to a training set with limited training samples and not the underlying probability distribution $\mu$, nor the manifold $\ccalM$, we can only aim to minimize the empirical risk in practice. Substantially, the goal of the GNN is to minimize the statistical risk $\PStatRiskGraph  = \min \limits_{\bbH\in \ccalH}  \StatRiskGraph (\bbH)$, with $\StatRiskGraph (\bbH)$ defined as
{\begin{align}\label{P:statistical_risk_graph}
    \StatRiskGraph (\bbH) :=   \mathbb{E}_{X_{N}\sim\mu^N}\left[ \ell\left( \bm\Phi(\bbH, \bbL_{N}, \bbx_{N}), \bby_{N}\right)\right],
\end{align}}where the expectation is taken with respect to $N$ randomly sampled points $X_{N}\sim \mu^{N}$.

The generalization gap of GNN composed of low-pass filters and normalized Lipschitz nonlinearities, which is in the same setting as Proposition \ref{prop:prob-diff} is defined to be 
\begin{align}
\label{eqn:ga}
    \GeneralizationGap = \PStatRiskGraph -\PEmpRiskGraph ,
\end{align}
which measures the difference between the optimal empirical risk and the optimal statistical risk of the GNN over the sampled graph from the manifold. {We note that by defining it in this way, the actual training loss, reflected by $P_E^*$, can be used to upper bound $P_S^*$. This provides a performance guarantee of the GNN over the whole dataset. }
Based on the convergence results that we have derived, we can bound the generalization gap in the following theorem. 
\begin{theorem}\label{thm:graph_geneneralization_gap}
    Suppose we have a GNN trained on $(\bbx_{N},\bbG_{N})$  with $N$ nodes sampled i.i.d. with measure $\mu$ over $\ccalM$, the generalization gap  GA is bounded in probability at least $1-\delta$ that, 
    \begin{align}
    {GA = \ccalO\left( \left(\frac{\log \frac{N}{\delta}}{N}\right)^{\frac{1}{d+4}} \right).}
\end{align}
\end{theorem}
We observe that the generalization gap $GA$ defined as \eqref{eqn:ga} decreases with the number of sampled points over the manifold $N$. 
Unlike previous results, we do not place any restrictions on the VC dimension or Rademacher complexities of the functions. We require that the graphs be sampled from the same underlying manifold, . Furthermore, the generalization gap increases with the manifold dimension $d$. That is to say, for a higher manifold complexity, we need to sample more points to achieve a satisfying generalization gap. {While in most practical scenarios, manifolds are often low-dimensional, especially those associated with physical systems or image manifolds, as demonstrated in \cite{talmon2015manifold, osher2017low, van2008visualizing}.}

\myparagraph{Remark 1} This conclusion can also be extended to multi-layer and multi-feature neural network architectures, as the neural network is cascaded by layers of filters and nonlinearities. The generalization gap may propagate across layers which indicates the increasing of the generalization gap of multi-layer and multi-feature GNNs with the size of the architecture, which we will further verify in Section \ref{sec:sim}.

{\myparagraph{Remark 2} The analysis of the generalization of GNNs on node classification tasks can be extended to graph classification tasks if we consider a linear combination of the node-wise outputs of the GNNs. The empirical risk $\EmpRiskGraph$ and the statistical risk $\StatRiskGraph$ are defined as the summation of loss functions over a set of graphs, with each graph sampled from a manifold within a manifold set.}

\subsection{Generalization Gap and Number of Nodes Relationship}

A salient consequence of Theorem \ref{thm:graph_geneneralization_gap} is that there exists a linear relationship in the logarithmic scale between the number of nodes in the graph and the generalization gap. To make this relationship apparent, we can compute the logarithm of the generalization gap as follows, 
\begin{align}\label{eqn:log_relationship}
    \texttt{log}(GA)\approx -a \texttt{log}(N)+b.
\end{align}
Where $b$ depends on $N$ in a $\log\log(N)$ relationship making it small in comparison with $\log(N)$. Therefore, Theorem \ref{thm:graph_geneneralization_gap} implies that the generalization gap decreases linearly (given that $a$ is a constant factor that depends on the parameters of the problem) with the logarithm of the number of nodes in the training set. 

Note that the bound obtained in \ref{thm:graph_geneneralization_gap} is an upper bound on the generalization gap. That is to say, the decrease in the generalization gap is at most linear with respect to the number of nodes in the training set. To empirically validate the result in Theorem \ref{thm:graph_geneneralization_gap}, we can compute the difference between the empirical and statistical risks of a trained GNN by training a GNN on a set of nodes and evaluating it on another set. In what follows, we carried out this experiment in two real-world datasets showcasing the validity of our theoretical result.

%% file: simulations.tex
\begin{figure*}
    \centering
    \includegraphics[width=\textwidth]{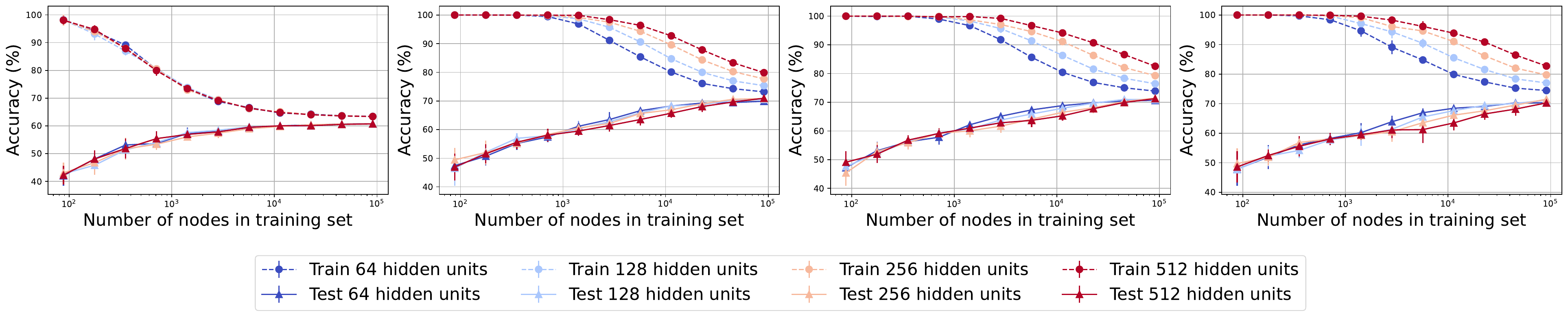}
    \includegraphics[width=\textwidth]{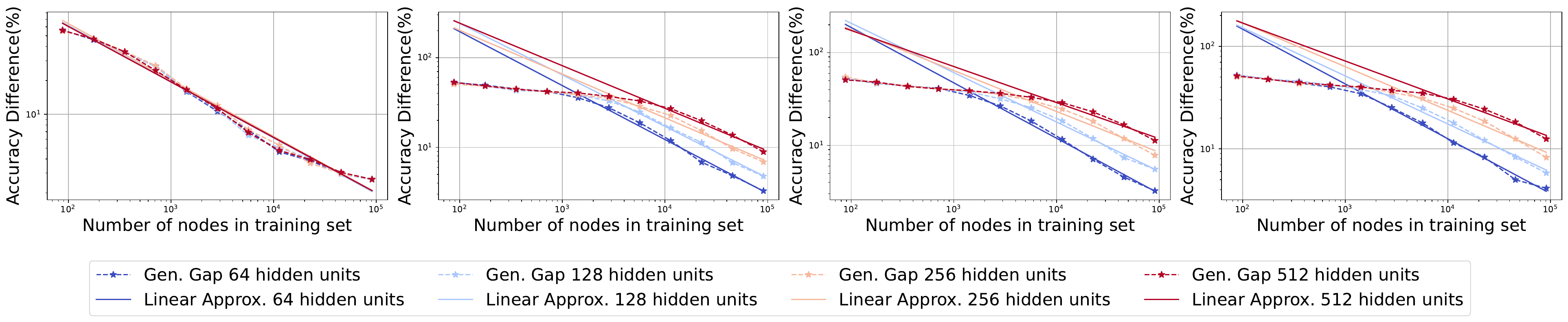}
    \begin{subfigure}{0.24\textwidth}
        \caption{One Layer}
    \end{subfigure}
    \begin{subfigure}{0.24\textwidth}
        \caption{Two Layers}
    \end{subfigure}
    \begin{subfigure}{0.24\textwidth}
        \caption{Three Layers}
    \end{subfigure}
    \begin{subfigure}{0.24\textwidth}
        \caption{Four Layers}
    \end{subfigure}
    \caption{{ GNN with $3$ layers and varying numbers of hidden units $\{64,128,256,512\}$ with a varying number of nodes in the training set trained in the
    Arxiv dataset. In the top row, we plot the accuracy in the training and testing sets as a function of the number of nodes in the training set for the Arxiv dataset. In the bottom row, we plot the accuracy difference and the best linear fit. For the linear approximation, we consider the points whose training accuracy is below $98\%$.  }}
    \label{fig:acc_arxiv}
\end{figure*}

\begin{figure*}
    \centering
    \includegraphics[width=\textwidth]{ 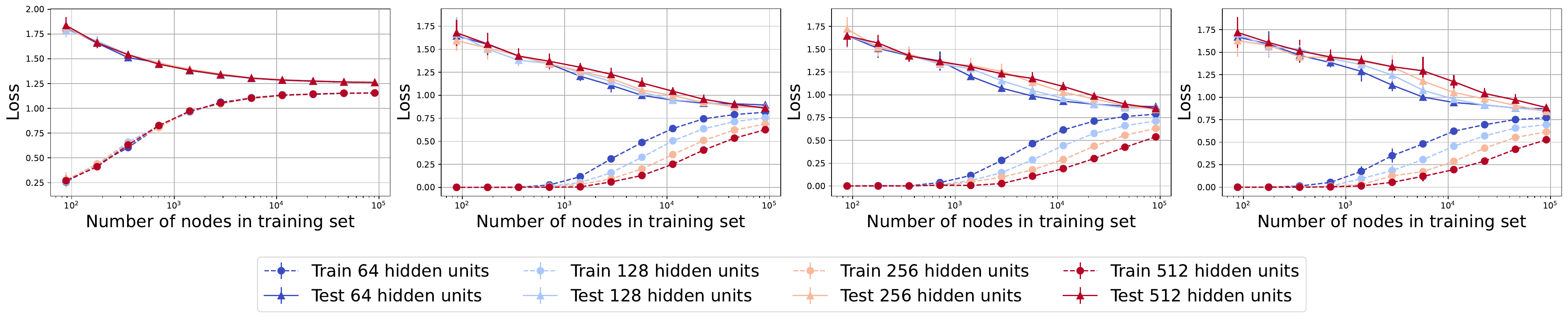}
    \includegraphics[width=\textwidth]{ 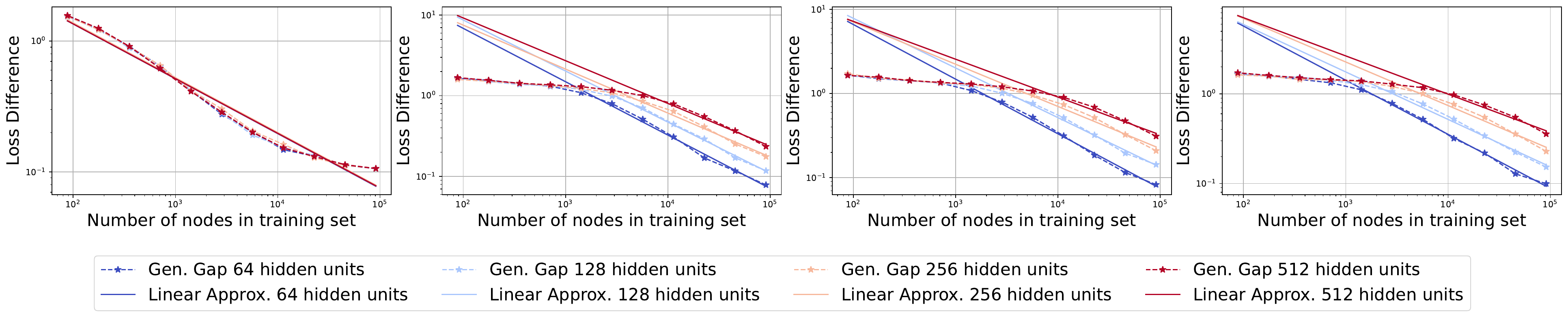}
    \begin{subfigure}{0.24\textwidth}
        \caption{One Layer}
    \end{subfigure}
    \begin{subfigure}{0.24\textwidth}
        \caption{Two Layers}
    \end{subfigure}
    \begin{subfigure}{0.24\textwidth}
        \caption{Three Layers}
    \end{subfigure}
    \begin{subfigure}{0.24\textwidth}
        \caption{Four Layers}
    \end{subfigure}
    \caption{{ GNN with $3$ layers and varying numbers of hidden units $\{64,128,256,512\}$ with a varying number of nodes in the training set trained in the
    Arxiv dataset. In the top row, we plot the loss in the training and testing sets as a function of the number of nodes in the training set for the Arxiv dataset. In the bottom row, we plot the accuracy difference and the best linear fit. For the linear approximation, we consider the points whose training accuracy is below $98\%$.  }}
    \label{fig:loss_arxiv}
\end{figure*}

\begin{figure*}
    \centering
    \includegraphics[width=\textwidth]{ 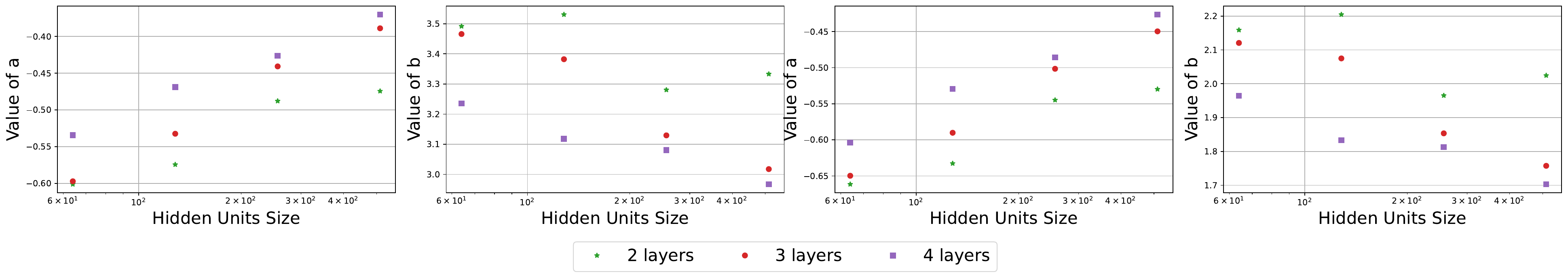}
    \begin{subfigure}{0.49\textwidth}
        \caption{Accuracy}
    \end{subfigure}
    \begin{subfigure}{0.49\textwidth}
        \caption{Loss}
    \end{subfigure}
    \caption{Values of slope (a) and point (b) corresponding to the linear fit ($a*\texttt{log}(N)+b$) of Figures \ref{fig:acc_arxiv} and \ref{fig:loss_arxiv}.  } 
    \label{fig:a_b_arxiv}
\end{figure*}

\begin{figure*}
    \centering
    \includegraphics[width=\textwidth]{ 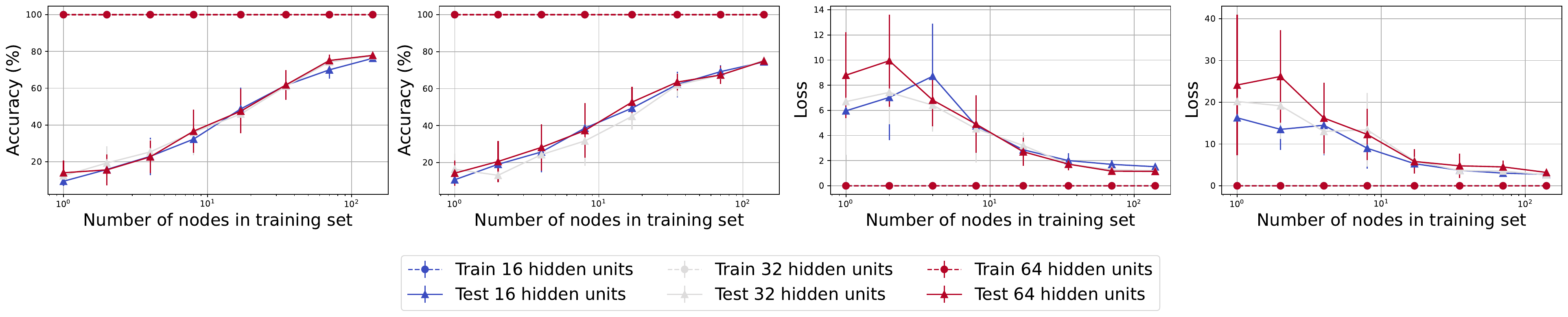}
    \includegraphics[width=\textwidth]{ 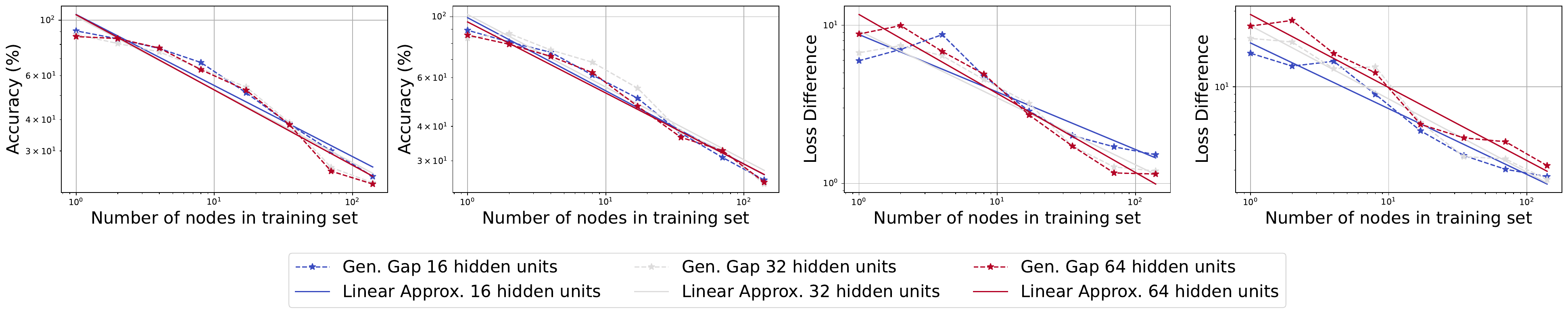}
    \begin{subfigure}{0.24\textwidth}
        \caption{Two Layers}
    \end{subfigure}
    \begin{subfigure}{0.24\textwidth}
        \caption{Three Layers}
    \end{subfigure}
    \begin{subfigure}{0.24\textwidth}
        \caption{Two Layers}
    \end{subfigure}
    \begin{subfigure}{0.24\textwidth}
        \caption{Three Layers}
    \end{subfigure}
    \caption{{ Cora dataset: GNN with $\{2,3\}$ layers and varying numbers of hidden units $\{16,32,64\}$ with a varying number of nodes in the training set trained in the Cora dataset. 
    In the top row, we plot the accuracy/loss in the training and testing sets as a function of the number of nodes in the training set for the Cora dataset. In the bottom row, we plot the accuracy difference and the best linear fit. }}
    \label{fig:acc_loss_cora}
\end{figure*}
\begin{figure*}
    \centering
    \includegraphics[width=\textwidth]{ 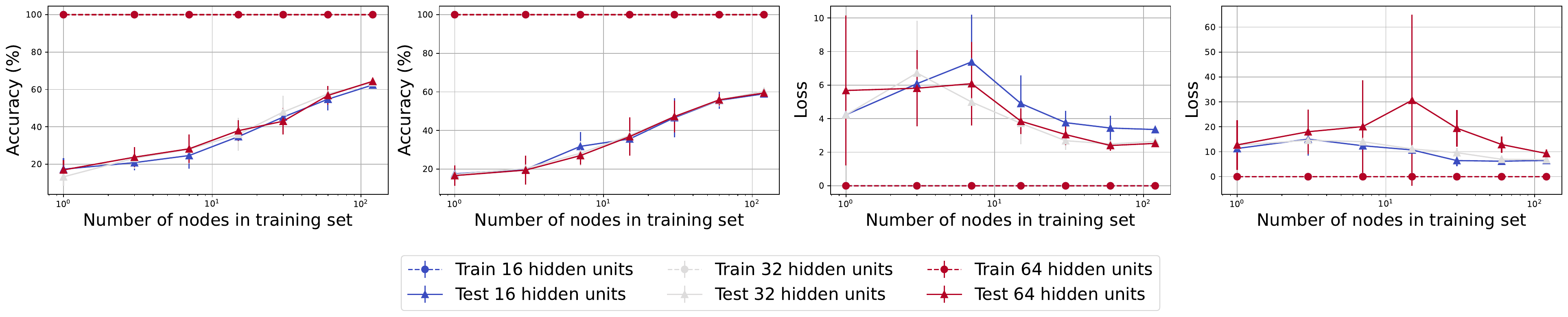}
    \includegraphics[width=\textwidth]{ 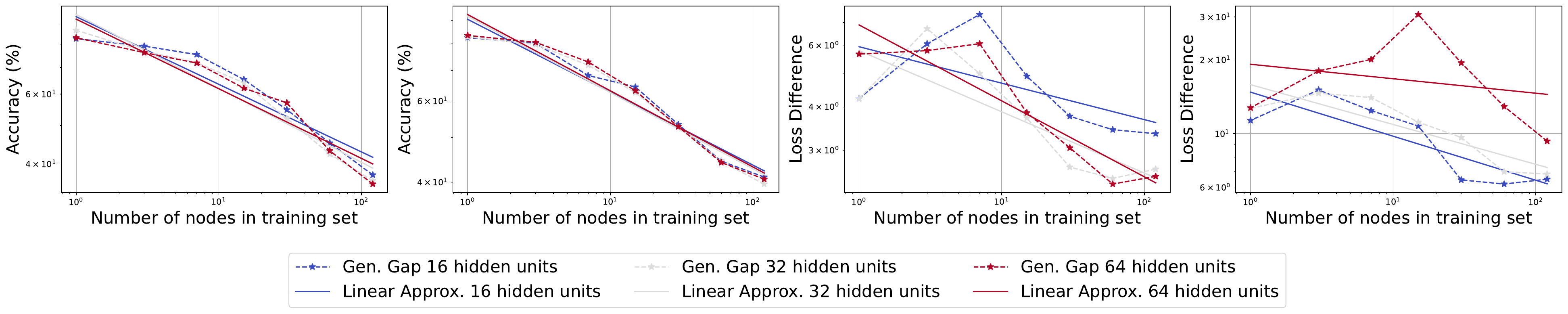}
    \begin{subfigure}{0.24\textwidth}
        \caption{Two Layers}
    \end{subfigure}
    \begin{subfigure}{0.24\textwidth}
        \caption{Three Layers}
    \end{subfigure}
    \begin{subfigure}{0.24\textwidth}
        \caption{Two Layers}
    \end{subfigure}
    \begin{subfigure}{0.24\textwidth}
        \caption{Three Layers}
    \end{subfigure}
    \caption{{CiteSeer dataset: GNN with $\{2,3\}$ layers and varying numbers of hidden units $\{16,32,64\}$ with a varying number of nodes in the training set trained in the
    CiteSeer dataset. In the top row, we plot the accuracy/loss in the training and testing sets as a function of the number of nodes in the training set for the Cora dataset. In the bottom row, we plot the accuracy difference and the best linear fit. }}
    \label{fig:acc_loss_citeseer}
\end{figure*}
\begin{figure*}
    \centering
    \includegraphics[width=\textwidth]{ 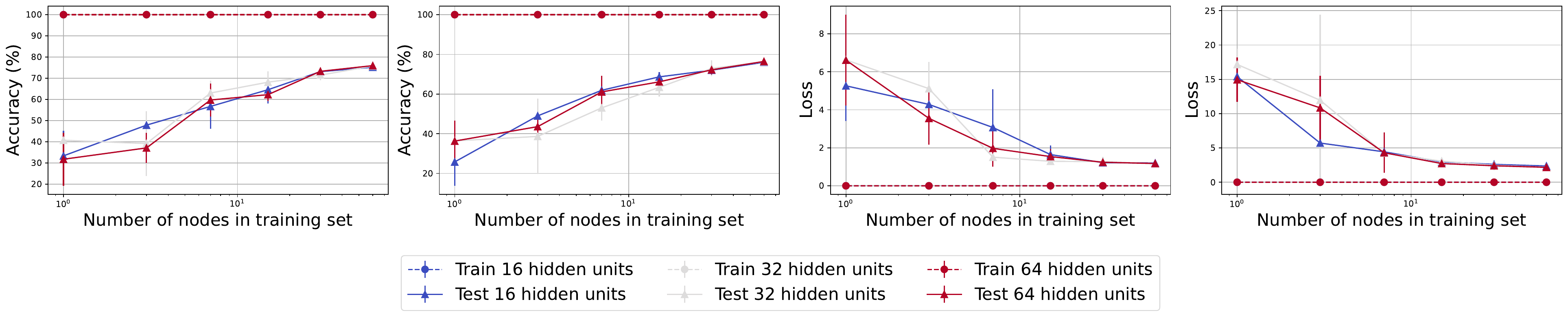}
    \includegraphics[width=\textwidth]{ 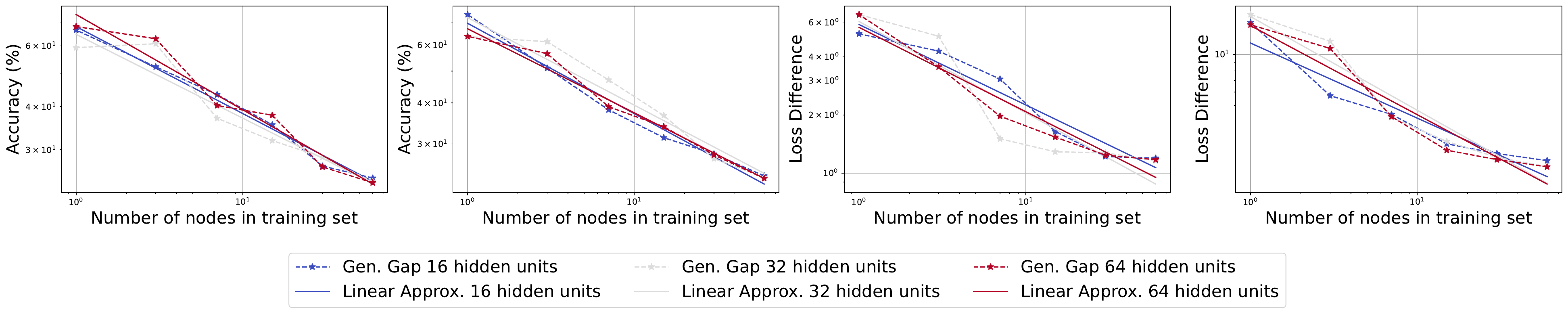}
    \begin{subfigure}{0.24\textwidth}
        \caption{Two Layers}
    \end{subfigure}
    \begin{subfigure}{0.24\textwidth}
        \caption{Three Layers}
    \end{subfigure}
    \begin{subfigure}{0.24\textwidth}
        \caption{Two Layers}
    \end{subfigure}
    \begin{subfigure}{0.24\textwidth}
        \caption{Three Layers}
    \end{subfigure}
    \caption{{PubMed dataset: GNN with $\{2,3\}$ layers and varying numbers of hidden units $\{16,32,64\}$ with a varying number of nodes in the training set trained in the
    PubMed dataset. In the top row, we plot the accuracy/loss in the training and testing sets as a function of the number of nodes in the training set for the Cora dataset. In the bottom row, we plot the accuracy difference and the best linear fit. }}
    \label{fig:acc_loss_pubmed}
\end{figure*}

In this section, we corroborate the proposed bound in real-world experiments. To this end, we learn a GNN in a dataset with a varying number of nodes. We denote the training set as the nodes seen in training and the testing set as the unseen nodes. Therefore, the training error is equivalent to the empirical risk \eqref{eqn:empirical_risk}, and the testing error is a way of estimating the statistical risk \eqref{eqn:statistical_risk}. We considered various layers and hidden units, and we repeated each experiment $10$ times. 

The task consists of predicting the class of a given node. We computed the training loss (measured in cross-entropy loss), and accuracy (measured in percentage) for each set. Additionally, we computed the difference between training and testing which corresponds to the generalization gap. 

All experiments used a \texttt{NVIDIA GeForce RTX 3090} GPU. We used \texttt{SGD} with a learning rate of $0.005$, no weight decay, and trained for $1000$ epochs. The reported value corresponds to the one obtained at the last iteration. In all cases we verified that the GNN converges to a stable value.

\subsection{Arxiv Dataset}
The Arxiv dataset represents the citation network between Computer Science Arxiv papers \cite{mikolov2013distributed,wang2020microsoft}. Every node in the dataset represents a paper, and each edge represents a citation between the two. The input signal for each node is a $128$ dimensional vector representing an embedding of the individual words in the title and abstract. The task is to classify the paper into $40$ subject areas of computer science. In all, this dataset is composed of $169,343$ nodes and $1,166,243$ edges. For this experiment, we trained GNNs with $\{1,2,3,4\}$ layers, and $\{64,128,256,512\}$ hidden units per layer.

\input{ table}

Our main goal in this experiment is to validate the theoretical bound derived in Theorem \ref{thm:graph_geneneralization_gap}. We evaluated this bound both for the accuracy (Figure \ref{fig:acc_arxiv}) as well as the loss (Figure \ref{fig:loss_arxiv}). In each plot's upper row, we show the training and testing accuracy/loss achieved by the trained GNN. Intuitively, when the number of nodes is small, the GNN can \textit{overfit} the training set,i.e. perfectly classify all the nodes in the training set. As the number of nodes increases, the GNN fails to overfit, and the training accuracy falls below $100\%$. In the lower, row of Figure \ref{fig:acc_arxiv} and Figure \ref{fig:loss_arxiv}, we plot both the generalization gap (the difference between training and testing accuracy/loss)
and the best linear fit for the points whose training accuracy is below $98\%$. Visually, we can see that the generalization gap follows a linear pattern in \texttt{log} scale. This is consistent with the theoretical result of Theorem \ref{thm:graph_geneneralization_gap}. Formally, we computed the Pearson correlation coefficient which measures the linear relationships between the data. In Table \ref{table:pearson} we show that the coefficient is always above $0.98$ both for accuracy as well as loss. This translates into a strong correlation, empirically validating the claim that the generalization gap follows a linear relationship in a logarithmic scale with respect to the number of nodes present in the empirical risk. 

We also qualitatively evaluate the relationship between the slope $a$ and point $b$ of the linear fit $a*\log(N)+b$ as a function of the number of layers, and hidden units per layer. In Figure \ref{fig:a_b_arxiv} we plot the values of the slope $a$ and point $b$ for accuracy and loss as a function of the number of hidden units in the GNN. In order to compute the value in logarithmic scale, we took the absolute value of the generalization gap. In the same figure, we include all the layers considered in the experiment. Intuitively, all slope values are negative, as the generalization gap decreases with the number of nodes in the dataset, a larger value of the slope $a$ indicates that the rate is \textit{slower}, whereas a smaller value (more negative) indicates that the generalization gap decreases faster. 
As can be seen in Figure \ref{fig:a_b_arxiv}, a salient conclusion is that as the number of hidden units increases, the generalization gap decreases more slowly, as the value of slope $a$ increases. The same is true for the number of layers. As the number of layers increases, the value of slope $a$ increases, and thus, the generalization gap decreases more slowly. This result is intuitive, given that the training error decreases more slowly as the size of the GNN is larger -- either more layers or more hidden units. As can be seen in Figure \ref{fig:acc_arxiv}, having more layers or more hidden units does not necessarily improve the testing loss. However, having only one layer in the GNN does not allow for a good generalization, as the accuracy in the testing set settles around $60\%$. In all, we show that the generalization gap indeed increases with the number of layers and hidden units in the GNN. 

\subsection{Cora Dataset}
The Cora dataset is a smaller graph than the Arxiv dataset. The task is also node label prediction, and in this case, there are $7$ classes. In this dataset, the input feature dimension is $1,433$, with $2,708$ nodes and $5,429$ edges \cite{yang2016revisiting}. In this case the experiments were done for $\{16,32,64\}$ hidden units and $\{2,3\}$ layers.

Analogous to the Arxiv dataset, our main goal in this experiment is to validate the theoretical bound derived in Theorem \ref{thm:graph_geneneralization_gap}. This experiment is important because we are in an overfitting regime for all numbers of nodes in the training set. That is to say, the GNNs always obtain $100\%$ accuracy in the training set as can be seen in the first two columns of Figure \ref{fig:acc_loss_cora}. This regime differs from the previous Arxiv dataset.

Visually, we verify that the generalization gap follows a linear decrease in log scale with respect to the number of nodes in the training set. This gap is driven by the increase in the testing accuracy/loss given that the training value remains constant. Formally, we can see in Table \ref{table:pearson} that the Pearson correlation coefficient is always above $0.96$ (except $2$ layers $64$ hidden units measured in the loss) which again shows a strong linear correlation between the number of nodes in the training set (in log scale) and the generalization gap. 

The conclusion of this experiment is positive, we empirically validated the claims that we put forward -- the generalization gap decreases linearly with the logarithm of the number of nodes. 
{  \subsection{CiteSeer Dataset}
The CiteSeer dataset has $3327$ nodes and $9228$ edges. There are $6$ classes, the feature dimension is $3703$, and we trained GNN with  $\{16,32,64\}$ hidden units and $\{2,3\}$ layers. 

In Figure \ref{fig:acc_loss_citeseer} we can see the generalization gap results which show a linear correlation in log-scale. In Table \ref{table:pearson} we show the Pearson correlation coefficient for these experiments. For Accuracy, the results show a strong linear correlation as our theory predicts. 

}
{ \subsection{PubMed Dataset}

The CiteSeer dataset has $19717$ nodes and $88651$ edges. There are $3$ classes and the feature dimension is $500$ dimensional, and we trained GNN with  $\{16,32,64\}$ hidden units and $\{2,3\}$ layers. 

In Figure \ref{fig:acc_loss_pubmed} we can see the generalization gap results which show a linear correlation in log-scale. In Table \ref{table:pearson} we show the Pearson correlation coefficient for these experiments. Both for Accuracy and Loss, the results show a strong linear correlation as our theory predicts. 
}

%% file: table.tex
\begin{table}[]
\centering
\begin{tabular}{|ccc|}
\hline
\multicolumn{1}{|c}{Layers} & \multicolumn{1}{c}{Hidden Units} & \begin{tabular}[c]{@{}c@{}}Pearson Correlation\\ Accuracy / Loss\end{tabular} \\ \hline\hline
\multicolumn{3}{|c|}{Arxiv Dataset}                                                                                                              \\ \hline\hline
$1$ & $64$  & $-0.991/-0.980$  \\ \hline
$1$ & $128$  & $-0.992/-0.980$  \\ \hline
$1$ & $256$  & $-0.994/-0.985$  \\ \hline
$1$ & $512$  & $-0.993/-0.981$  \\ \hline
$2$ & $64$  & $-0.997/-0.997$  \\ \hline
$2$ & $128$  & $-0.997/-0.998$  \\ \hline
$2$ & $256$  & $-0.991/-0.993$  \\ \hline
$2$ & $512$  & $-0.991/-0.994$  \\ \hline
$3$ & $64$  & $-0.997/-0.998$  \\ \hline
$3$ & $128$  & $-0.994/-0.996$  \\ \hline
$3$ & $256$  & $-0.982/-0.988$  \\ \hline
$3$ & $512$  & $-0.982/-0.987$  \\ \hline
$4$ & $64$  & $-0.997/-0.998$  \\ \hline
$4$ & $128$  & $-0.993/-0.994$  \\ \hline
$4$ & $256$  & $-0.983/-0.989$  \\ \hline
$4$ & $512$  & $-0.984/-0.988$  \\ \hline\hline
\multicolumn{3}{|c|}{Cora Dataset}  \\ \hline\hline
$2$ & $16$  & $-0.980/-0.925$   \\ \hline
$2$ & $32$  & $-0.969/-0.966$   \\ \hline
$2$ & $64$  & $-0.969/-0.978$   \\ \hline
$3$ & $16$  & $-0.989/-0.975$   \\ \hline
$3$ & $32$  & $-0.968/-0.973$   \\ \hline
$3$ & $64$  & $-0.987/-0.976$   \\ \hline \hline
\multicolumn{3}{|c|}{{  CiteSeer Dataset}}  \\ \hline\hline
$2$ & $16$  & $-0.952/-0.592$   \\ \hline
$2$ & $32$  & $-0.971/-0.798$   \\ \hline
$2$ & $64$  & $-0.953/-0.912$   \\ \hline
$3$ & $16$  & $-0.972/-0.835$   \\ \hline
$3$ & $32$  & $-0.972/-0.877$   \\ \hline
$3$ & $64$  & $-0.970/-0.256$   \\ \hline \hline
\multicolumn{3}{|c|}{{  PubMed Dataset}} \\ \hline\hline
$2$ & $16$  & $-0.994/-0.972$   \\ \hline
$2$ & $32$  & $-0.958/-0.914$   \\ \hline
$2$ & $64$  & $-0.976/-0.972$   \\ \hline
$3$ & $16$  & $-0.991/-0.955$   \\ \hline
$3$ & $32$  & $-0.971/-0.966$   \\ \hline
$3$ & $64$  & $-0.989/-0.963$   \\ \hline
\end{tabular}
\caption{{Pearson Correlation Coefficients for the  linear fit computed on Figures \ref{fig:acc_arxiv}, \ref{fig:loss_arxiv}, \ref{fig:acc_loss_cora}, \ref{fig:acc_loss_citeseer}, and \ref{fig:acc_loss_pubmed} for the Arxiv, Cora, CiteSeer and PubMed datasets respectively. For the Arxiv dataset the coefficient was computed using the points corresponding to training accuracy below $98\%$. For the Cora dataset, all points were used.  }}
\label{table:pearson}
\end{table}

%% file: conclusion.tex

In this paper, we implemented a manifold convolutional neural network as a limit model for graph neural networks. We considered the case when graphs are constructed from samples from the manifold. We showed that when a graph is constructed with the i.i.d. randomly sampled points over the manifold, the GNN converges to the MNN with a non-asymptotic rate both in probability and in expectation. We utilized these convergence results as a stepping stone towards proving an upper bound on the generalization gap of GNNs.  Formally, we showed that the difference between the statistical and empirical risks decreases with the number of sampled points from the manifold while it increases with the dimension of the manifold.
Furthermore, our bound admits a linear interpretation in a logarithmic scale which can be validated empirically.
Through experiments, we verified our convergence and generalization results numerically with two real-world datasets. We trained GNNs on a varying number of nodes and compared their generalization gaps. We empirically verified that the theoretically proposed rate matches the one seen in practice for a variety of layers and hidden dimensions. We also computed the Pearson Correlation coefficient to validate the linear relationships and for the most part, found values always above $0.9$, which indicates a strong linear relationship.

%% file: appendix.tex
\subsection{Proof of Proposition \ref{prop:prob-diff}}
\label{app:prop1}
To prove the convergence result of the output of GNN $\bm\Phi(\bbH,\bbL_N,\bbx_N)$ to MNN $\bm\Phi(\bbH,\ccalL_\rho, f)$, we first need to import a proposition that states the spectrum convergence of graph Laplacian $\bbL_N$ to the weighted Laplace operator $\ccalL_\rho$. 

 \begin{proposition}{\cite[Theorem~2.4, Theorem~2.6]{calder2022improved}}\label{thm:converge-spectrum-sparse}
Let $\ccalM\subset \reals^\mathsf{N}$ of measure $\mu$ with probability density function $\rho$ be equipped with a weighted Laplace operator $\ccalL_\rho$ as defined in \eqref{eqn:weight-Laplace}, whose eigendecomposition is given by $\{\lambda_i,\bm\phi_i\}_{i=1}^\infty$. Let $\bbL_N$ be the discrete graph Laplacian of graph weights defined as \eqref{eqn:edge-weight}, with spectrum $\{\lambda_{i,N},\bm\phi_{i,N}\}_{i=1}^N$.
Fix $K\in \mathbb{N}^+$ and assume that  $\epsilon=\epsilon(N)\geq \left({\log(CN/\delta)}/{N}\right)^{1/(d+4)}$  
Then, with probability at least $1-\delta$, we have 
\begin{equation}
    |\lambda_i-\lambda_{i,N} |\leq C_{\ccalM,1} \lambda_i{\epsilon},  \quad 
    \|a_i\bm\phi_{i,N} -\bm\phi_i\|\leq C_{\ccalM,2} \frac{\lambda_i}{\theta_i} {\epsilon} ,
\end{equation}
with $a_i\in\{-1,1\}$ for all $i<K$ and $\theta_i$ the eigengap of $\ccalL$, i.e., $\theta_i=\min\{\lambda_i-\lambda_{i-1},\lambda_{i+1}-\lambda_{i}\}$. The constants $ C_{\ccalM,1}$, $ C_{\ccalM,2}$ depend on $d$ and the volume of $\ccalM$.
\end{proposition}
The inner product of signals $f, g\in L^2(\ccalM)$ is defined as 
\begin{equation}\label{eqn:innerproduct}
    \langle f,g \rangle_{\ccalM}=\int_\ccalM f(x)g(x) \text{d}\mu(x), 
\end{equation}
where $\text{d}\mu(x)$ is the volume element with respect to the measure $\mu$ over $\ccalM$. The norm of the manifold signal $f$ is subsequently defined as 
\begin{equation}\label{eqn:manifold_norm}
    \|f\|^2_{\ccalM}={\langle f,f \rangle_{\ccalM}}.
\end{equation}
Because $\{x_1, x_2,\cdots,x_N\}$ is a set of randomly sampled points from $\ccalM$, based on Theorem 19 in \cite{von2008consistency} we can claim that
\begin{equation}
   \left|\langle \bbP_N f,\bm\phi_{i } \rangle   -\langle f,\bm\phi_i\rangle_\ccalM\right| = O\left(\sqrt{\frac{\log (1/\delta)}{N}}\right).
\end{equation}
This also indicates that 
\begin{equation}
\label{eqn:diff-sample}
   \left|\|\bbP_N f \|^2-\|f\|^2_\ccalM\right| = O\left(\sqrt{\frac{\log (1/\delta)}{N}}\right),
\end{equation}
which indicates $\|\bbP_N f\|=\|f\|_\ccalM + O((\log (1/\delta)/N)^{1/4})$.
As GNN and MNN are composed of graph filters and manifold filters, respectively, we start with deriving the output difference of a graph filter $\bbh(\bbL_N)$ and a manifold filter $\bbh(\ccalL_\rho)$. The input of GNN is a sampled version of bandlimited manifold signal $f$. We first write out the difference of filter outputs as
{\allowdisplaybreaks
 \begin{align}
    &\nonumber \|\bbh(\bbL_N)\bbP_N f - \bbP_N \bbh(\ccalL_\rho) f\|\leq \\
    & \label{eqn:1-1} \Bigg\| \sum_{i=1}^N \hat{h}(\lambda_{i,N}) \langle \bbP_N f,\bm\phi_{i,N} \rangle \bm\phi_{i,N}  - \sum_{i=1}^M \hat{h}(\lambda_i)\langle f,\bm\phi_i\rangle_{\ccalM} \bbP_N \bm\phi_i  \Bigg\|
     \\ 
     &\nonumber \leq  \Bigg\| \sum_{i=1}^M \hat{h}(\lambda_{i,N}) \langle \bbP_N f,\bm\phi_{i,N} \rangle \bm\phi_{i,N} - \sum_{i=1}^M \hat{h}(\lambda_i) \langle \bbP_Nf,\bm\phi_{i}  \rangle_{\ccalM}\bm\phi_{i}\\
     &\qquad \qquad \qquad \label{eqn:1-2} + \sum_{i=M+1}^N \hat{h}(\lambda_{i,N}) \langle \bbP_N f,\bm\phi_{i,N} \rangle \bm\phi_{i,N} \Bigg\| \\
     &\nonumber \leq \Bigg\| \sum_{i=1}^M \hat{h}(\lambda_{i,N}) \langle \bbP_N f,\bm\phi_{i,N} \rangle \bm\phi_{i,N} - \sum_{i=1}^M \hat{h}(\lambda_i) \langle \bbP_Nf,\bm\phi_{i}  \rangle_{\ccalM}\bm\phi_{i}\Bigg\| \\
     & \qquad \qquad \qquad \label{eqn:1-3}+ \left\|\sum_{i=M+1}^N \hat{h}(\lambda_{i,N}) \langle \bbP_N f,\bm\phi_{i,N} \rangle \bm\phi_{i,N}\right\|.
 \end{align}}
 From \eqref{eqn:1-1} to \eqref{eqn:1-2}, we decompose the spectral components of the output of graph filter as the summation of the first $M$ items and the rest $N-M$ items. From \eqref{eqn:1-2} to \eqref{eqn:1-3}, we use the triangle inequality to decompose the filter output difference as two terms.
 
 The first part of \eqref{eqn:1-3} can be decomposed by subtracting and adding an intermediate term $\sum_{i=1}^M\hat{h}(\lambda_i) \langle \bbP_Nf,\bm\phi_{i,N} \rangle \bm\phi_{i,N}$. From \eqref{eqn:2-1} to \eqref{eqn:conv-1}, we use the triangle inequality.
 \begin{align}
     & \nonumber \left\| \sum_{i=1}^M \hat{h}(\lambda_{i,N}) \langle \bbP_N f,\bm\phi_{i,N}  \rangle \bm\phi_{i,N}  - \sum_{i=1}^M \hat{h}(\lambda_i)\langle f,\bm\phi_i\rangle_{\ccalM} \bbP_N \bm\phi_i  \right\|
     \\ &\nonumber \leq  \Bigg\| \sum_{i=1}^M  \hat{h}(\lambda_{i,N}) \langle \bbP_Nf,\bm\phi_{i,N} \rangle \bm\phi_{i,N}-  \sum_{i=1}^M\hat{h}(\lambda_i) \langle \bbP_Nf,\bm\phi_{i,N} \rangle \bm\phi_{i,N}  \\
     & \label{eqn:2-1} +\sum_{i=1}^M \hat{h}(\lambda_i)  \langle \bbP_N f,\bm\phi_{i,N} \rangle  \bm\phi_{i,N} - \sum_{i=1}^M \hat{h}(\lambda_i)\langle f,\bm\phi_i \rangle_{\ccalM} \bbP_N \bm\phi_i   \Bigg\|.
     \\
     &\nonumber \leq  \left\| \sum_{i=1}^M \left(\hat{h}(\lambda_{i,N})- \hat{h}(\lambda_i) \right) \langle \bbP_Nf,\bm\phi_{i,N} \rangle \bm\phi_{i,N} \right\| \\
     &  +\left\| \sum_{i=1}^M \hat{h}(\lambda_i)\left( \langle \bbP_N f,\bm\phi_{i,N} \rangle  \bm\phi_{i,N} - \langle f,\bm\phi_i \rangle_{\ccalM} \bbP_N \bm\phi_i \right)  \right\|.\label{eqn:conv-1}
 \end{align}
In equation \eqref{eqn:conv-1}, we can observe that the first part relies on the difference of eigenvalues, while the second part depends on the eigenvector difference. In the following, we bound these two parts based on Proposition \ref{thm:converge-spectrum-sparse} as follows. 

The first term in \eqref{eqn:conv-1} can be decomposed based on the triangle inequality and Cauchy-Schwartz inequality. We note that the eigenvectors are orthonormal which indicates that $\|\bm\phi_{i,N}\| = 1$.
\begin{align}
   &\nonumber \left\| \sum_{i=1}^M (\hat{h}(\lambda_{i,n} ) - \hat{h}(\lambda_i)) \langle \bbP_N f,\bm\phi_{i,N} \rangle  \bm\phi_{i,N}  \right\|  \\
   & \leq \sum_{i=1}^M \left|\hat{h}(\lambda_{i,N} )-\hat{h}(\lambda_i)\right| |\langle \bbP_N f,\bm\phi_{i,N}  \rangle |\label{eqn:p1} \\
   &\leq \|\bbP_N f\|  \sum_{i=1}^M C_{\ccalM,1}\epsilon \lambda_i^{-d}  \leq  \|\bbP_N f\|C_{\ccalM,1}\epsilon  \sum_{i=1}^M  i^{-2} \label{eqn:p2}\\
   &\leq \left( \|f\|_\ccalM+ \left(\frac{\log (1/\delta)}{N}\right)^{\frac{1}{4}}\right)C_{\ccalM,1}\epsilon \frac{\pi^2}{6} := A_1(N).
\end{align} 
From \eqref{eqn:p1} to \eqref{eqn:p2}, we use the fact that the magnitude of the filters scale with $\lambda^{-d}$. 
In \eqref{eqn:p2}, we implement Weyl's law \cite{arendt2009weyl} which indicates that eigenvalues of Laplace operator scales with the order of $i^{2/d}$. The last inequality comes from the fact that $\sum_{i=1}^\infty i^{-2}=\frac{\pi^2}{6}$.

The second term in \eqref{eqn:conv-1} can be written as subtracting and adding an intermediate-term $\sum_{i=1}^M \hat{h}(\lambda_i)\langle \bbP_Nf,\bm\phi_{i,N}  \rangle  \bbP_N\bm\phi_i$ and decomposed with the triangle inequality as follows.
{\allowdisplaybreaks
\begin{align}
  & \nonumber \Bigg\| \sum_{i=1}^M \hat{h}(\lambda_i)\left( \langle \bbP_Nf,\bm\phi_{i,N}  \rangle \bm\phi_{i,N}- \langle f,\bm\phi_i \rangle_{\ccalM} \bbP_N \bm\phi_i\right)  \Bigg\|\\
   & \leq \nonumber \Bigg\|  \sum_{i=1}^M \hat{h}(\lambda_i)  \left(\langle \bbP_N f,\bm\phi_{i,N} \rangle \bm\phi_{i,N}   - \langle \bbP_Nf,\bm\phi_{i,N}  \rangle  \bbP_N\bm\phi_i\right)\Bigg\|\\
   &\label{eqn:term1}+ \left\| \sum_{i=1}^M  \hat{h}(\lambda_i) \left(\langle \bbP_N f,\bm\phi_{i,N} \rangle  \bbP_N\bm\phi_i -\langle f,\bm\phi_i\rangle_\ccalM \bbP_N\bm\phi_i \right) \right\|
\end{align}}
The first term in \eqref{eqn:term1} can be bounded with the triangle inequality and Cauchy-Schwartz inequality as
\begin{align}
& \nonumber \left\|  \sum_{i=1}^M \hat{h}(\lambda_i) \left(\langle \bbP_N f,\bm\phi_{i,N} \rangle \bm\phi_{i,N}  - \langle \bbP_Nf,\bm\phi_{i,N} \rangle_{\ccalM} \bbP_N\bm\phi_i\right)\right\|\\
& \label{eqn:3-1} \leq \sum_{i=1}^{M} \left|\hat{h}(\lambda_i)\right|\|\bbP_N f\|\|\bm\phi_{i,N} - \bbP_N\bm\phi_i\|\\
&\label{eqn:3-2} \leq \sum_{i=1}^M (\lambda_i^{-d}) \frac{C_{\ccalM,2}\epsilon}{\theta_i}\left(\|f\|_\ccalM + \left(\frac{\log (1/\delta)}{N}\right)^{\frac{1}{4}}\right)\\
&\label{eqn:3-3} \leq C_{\ccalM,2}\epsilon \frac{\pi^2}{6} \max_{i=1,\cdots,M}{\theta_i^{-1}}\left(\|f\|_\ccalM + \left(\frac{\log (1/\delta)}{N}\right)^{\frac{1}{4}}\right)\\
&:= A_2(M,N).
\end{align}
From \eqref{eqn:3-1} to \eqref{eqn:3-2}, we relies on the magnitude assumption of the filter response function in Definition \ref{def:filter} and inequality \eqref{eqn:diff-sample} as well as the eigenvector gap in Proposition \ref{thm:converge-spectrum-sparse}.  From \eqref{eqn:3-2} to \eqref{eqn:3-3}, we use the Weyl's law which states the order $\lambda_i\sim i^{2/d}$ and the summation $\sum_{i=1}^\infty i^{-2}=\pi^2/6$. 

The second term in \eqref{eqn:term1} with triangle inequality can be written as
\begin{align}
     & \nonumber \Bigg\| \sum_{i=1}^M \hat{h}(\lambda_{i,N} ) (\langle \bbP_N f,\bm\phi_{i,N}\rangle  \bbP_N\bm\phi_i -\langle f,\bm\phi_i\rangle_\ccalM \bbP_N\bm\phi_i ) \Bigg\| \\
   &\label{eqn:4-1}\leq \sum_{i=1}^M \left|\hat{h}(\lambda_{i,N}) \right| \left|\langle \bbP_N f,\bm\phi_{i,N} \rangle   -\langle f,\bm\phi_i\rangle_\ccalM\right|\|\bbP_N\bm\phi_i\|\\
   &\label{eqn:4-2} \leq \sum_{i=1}^M (\lambda_i^{-d})\sqrt{\frac{\log (1/\delta)}{N}} \left(1+\left(\frac{\log (1/\delta)}{N}\right)^{\frac{1}{4}}\right)\\
   &\label{eqn:4-3} \leq \frac{\pi^2}{6}\sqrt{\frac{\log (1/\delta)}{N}} \left(1+\left(\frac{\log (1/\delta)}{N}\right)^{\frac{1}{4}}\right):=A_3(N)
\end{align}
From \eqref{eqn:4-1} to \eqref{eqn:4-2}, it relies on the magnitude of filter in Definition \ref{def:filter} and inequality \eqref{eqn:diff-sample} of the norm difference caused by sampling. From \eqref{eqn:4-2} to \eqref{eqn:4-3}, we use the order of eigenvalues based on Weyl's law and the summation inequality as previous.

The second term in \eqref{eqn:1-3} can be bounded with triangle inequality and Cauchy-Schwartz inequality as
\begin{align}
   &\nonumber  \left\|\sum_{i=M+1}^N \hat{h}(\lambda_{i,N}) \langle \bbP_N f,\bm\phi_{i,N} \rangle \bm\phi_{i,N}\right\|\\
   &\label{eqn:5-1}\leq \sum_{i=M+1}^N \left|\hat{h}(\lambda_{i,N})\right|\|\bbP_N f\| 
   \|\bm\phi_{i,N}\|\\
    &\label{eqn:5-2}\leq \sum_{i=M+1}^N (\lambda_{i,N}^{-d})\left(\|f\|_\ccalM+\left(\frac{\log (1/\delta)}{N}\right)^{\frac{1}{4}}\right)\\
    &\label{eqn:5-3}\leq \sum_{i=M+1}^N (\lambda_{i,N}^{-d+1}) 
    \|f\|_\ccalM
    \\&\label{eqn:5-4}\leq  (1 +C_{\ccalM,1}\epsilon)^{-d} \sum_{i=M+1}^\infty (\lambda_i^{-d})  \|f\|_\ccalM\\
    &\label{eqn:5-5}\leq  M^{-1}\|f\|_\ccalM:= A_4(M).
\end{align}
From \eqref{eqn:5-1} to \eqref{eqn:5-2}, it follows a similar way as we have derived from \eqref{eqn:4-1} to \eqref{eqn:4-2}. From \eqref{eqn:5-2} to \eqref{eqn:5-3}, we scale by taking the summation to infinity. From \eqref{eqn:5-3} to \eqref{eqn:5-4}, it depends on the eigenvalue gap in Proposition \ref{thm:converge-spectrum-sparse}. From \eqref{eqn:5-4} to \eqref{eqn:5-5}, we use the order of the Weyl's law and the inequality $\sum_{i=M}^\infty i^{-2}\leq 1/M$.

We note that the bound is made up by adding all the decomposed terms $A_1( N)+A_2(M,N)+A_3( N)+A_4(M)$, related to the bandwidth of manifold signal $M$ and the number of sampled points $N$. As $\epsilon$ scales with the order $\left(\frac{\log(CN/\delta)}{N} \right)^{\frac{1}{d+4}}$. This makes the bound scale with the order
{\allowdisplaybreaks
\begin{align}
    &\nonumber \|\bbh(\bbL_N)\bbP_N f - \bbP_N \bbh(\ccalL_\rho) f\| \\
    &\nonumber \leq C_1 \left(\frac{\log\frac{C_1N}{\delta}}{N} \right)^{\frac{1}{d+4}} +  C_2\left(\frac{\log\frac{C_1N}{\delta}}{N} \right)^{\frac{1}{d+4}} \theta_M^{-1} \\
    &\qquad \qquad \qquad \qquad  + C_3 \sqrt{\frac{\log(1/\delta)}{N}} + C_4 M^{-1},
\end{align}}
with $C_1 = C_{\ccalM,1}\frac{\pi^2}{6}\|f\|_\ccalM$, $C_2 = C_{\ccalM,2}\frac{\pi^2}{6}$, $C_3 =\frac{\pi^2}{6}$ and $C_4 = \|f\|_\ccalM$.

With the normalized Lipschitz nonlinearities in Assumption \ref{ass:activation}, we have the single layer output difference bounded. From \eqref{eqn:6-1} to \eqref{eqn:6-2}, it relies on the point-wise operation of nonlinearity functions. From \eqref{eqn:6-2} to \eqref{eqn:6-3}, we use the Assumption \ref{ass:activation} that nonlinearities are normalized Lipschitz.
 \begin{align}
    &\nonumber\|\bm\Phi(\bbH,\bbL_N, \bbP_N f) - \bbP_N\bm\Phi(\bbH, \ccalL_\rho ,f ) \| \\
    & \label{eqn:6-1}= \| \sigma(\bbh(\bbL_N)\bbP_N f) - \bbP_N\sigma(\bbh(\ccalL_\rho) f) \| \\
    &\label{eqn:6-2} = \| \sigma(\bbh(\bbL_N)\bbP_N f) - \sigma(\bbP_N \bbh(\ccalL_\rho) f) \|
 \\&\label{eqn:6-3}  \leq \|\bbh(\bbL_N)\bbP_N f - \bbP_N \bbh(\ccalL_\rho) f\|\\
\nonumber
&  \label{eqn:6-4}\leq   C_1 \left(\frac{\log\frac{C_1N}{\delta}}{N} \right)^{\frac{1}{d+4}} +  C_2\left(\frac{\log\frac{C_1N}{\delta}}{N} \right)^{\frac{1}{d+4}} \theta_M^{-1} \\
    &\qquad \qquad \qquad \qquad\qquad   + C_3 \sqrt{\frac{\log(1/\delta)}{N}} + C_4 M^{-1}.
    \end{align}

As $N$ goes to infinity, for every $\delta >0$, there exists some $M_0$, such that for all $M>M_0$ it holds that $A_4(M)\leq \delta/2$. There also exists $n_0$, such that for all $N>n_0$, it holds that $A_1(N)+A_2(M_0, N)+A_3(N)\leq \delta/2$. We can conclude that the summations converge as $N$ goes to infinity.

\subsection{Proof of Corollary \ref{cor:expect-diff}}
\label{app:cor2}
We focus on deriving this expectation depending on the difference bound in probability in \eqref{eqn:prob-diff}. We denote $\hat \bby_N = \bm\Phi(\bbH,\bbL_N, \bbP_N f)$ and $\hat g = \bm\Phi(\bbH,\ccalL_\rho, f)$. We have 
\begin{align}
    &\nonumber \mathbb{P}\Bigg( \|\hat\bby_N - \bbP_N\hat g\|\leq (C_1+C_2\theta_M^{-1}) N^{-\frac{1}{d+4}}\left(\log\frac{N}{\delta} \right)^{\frac{1}{d+4}} \\
    &  \qquad \qquad +C_3 N^{-\frac{1}{2}} \sqrt{
    \log \frac{1}{\delta}} + C_4 M^{-1} 
     \Bigg) \geq 1-2\delta,
\end{align}
with $N \geq \log\frac{C}{\delta}$. We denote $k^2 = \log 1/\delta$, i.e. $\delta=e^{-k^2}$. We calculate the expectation value of the difference between $\hat{\bby}_N$ and $\bbP_N \hat{g}$ by decomposing the event space as follows.
{\allowdisplaybreaks
\begin{align}
   \nonumber  &\mathbb{E} [\|\hat\bby_N - \bbP_N\hat g\| ] \leq \sum_{k=0}^{\sqrt{N/C }} \mathbb{P}\Bigg(C_3 N^{-\frac{1}{2}} k +   C_4 M^{-1} \\ \nonumber  &(C_1+C_2\theta_M^{-1})) N^{-\frac{1}{d+4}} (\log N +k^2)^{\frac{1}{d+4}}
       \leq \|\hat\bby_N - \bbP_N\hat g\|
   \\
   &  \nonumber\leq (C_1+C_2\theta_M^{-1}) N^{-\frac{1}{d+4}} (\log N + (k+1)^2)^{\frac{1}{d+4}} \\& \nonumber +C_3 N^{-\frac{1}{2}} (k+1) + C_4 M^{-1}  \Bigg) \Bigg(C_3 N^{-\frac{1}{2}} (k+1) + C_4 M^{-1}
   \\&\nonumber +(C_1+C_2\theta_M^{-1}) N^{-\frac{1}{d+4}} (\log N + (k+1)^2)^{\frac{1}{d+4}}  \Bigg) + \nonumber \\ \nonumber &\sum_{k=\sqrt{N/C }}^\infty \bar M \mathbb{P}\Bigg( (C_1+C_2\theta_M^{-1})) N^{-\frac{1}{d+4}} (\log N +k^2)^{\frac{1}{d+4}}   \\
   & \nonumber +C_3 N^{-\frac{1}{2}} k +C_4 M^{-1}   \leq \|\hat\bby_N - \bbP_N\hat g\|_2\leq C_4 M^{-1}+
   \\
   &  \nonumber (C_1+C_2\theta_M^{-1}) N^{-\frac{1}{d+4}} (\log N + (k+1)^2)^{\frac{1}{d+4}}+C_3 N^{-\frac{1}{2}} (k+1) \Bigg).
\end{align}}
The upper bound of $\|\hat\bby_N - \bbP_N\hat g\|$ can be derived with the norm of the output function of GNN when the GNN contains a single layer 
\begin{align}
    \|\hat\bby_N - \bbP_N\hat g\| &\leq \| \bm\Phi(\bbH, \bbL,\bbP_N f)\|  +\|\bbP_N \bm\Phi(\bbH,\ccalL, f)\|  \\
    &\leq 2 \|\bbP_N f\| =\bar M.
\end{align}
Therefore the probability is zero when $ (C_1+C_2\theta_M) N^{-\frac{1}{d+4}} (k+1)^{\frac{2}{d+4}} + C_3 N^{-\frac{1}{2}} (k+1)  > \bar M = 2 \|\bbP_N f\|$, i.e. $k>\sqrt{N}\bar M$. Then we have 
{\allowdisplaybreaks
\begin{align}
    &\nonumber \mathbb{E}_\mu^N [\|\hat\bby_N - \bbP_N\hat g\|]\\
    &\nonumber \leq \sum_{k=0}^{\sqrt{N/C}} 2 e^{-k^2} \Bigg(C_3 N^{-\frac{1}{2}} (k+1) + C_4 M^{-1}  + \\
    &\qquad \qquad \qquad \nonumber(C_1+C_2\theta_M^{-1}) N^{-\frac{1}{d+4}} (\log N + (k+1)^2)^{\frac{1}{d+4}} \Bigg)  \\
&\qquad \qquad \qquad \quad \qquad \qquad \qquad  + \sum_{k=\sqrt{N/C}}^{\sqrt{N}\bar M} \bar M 2e^{-N/C} \\
    & \label{eqn:8-3}\leq C' N^{-\frac{1}{d+4}} + C''N^{-\frac{1}{2}} +C'''\left(\frac{\log N}{N}\right)^{\frac{1}{d+4}} + \bar M^2 e^{-N/C}\sqrt{N}.
\end{align}}
The last inequality is bounded by 
\begin{align}
   &\nonumber  \int_0^\infty 2e^{-k^2}(k+1)\text{d}k, \int_0^\infty 2e^{-k^2}(k+1)^{\frac{2}{d+4}}\text{d}k,\\
   &\qquad \qquad \qquad \qquad \int_0^\infty 2e^{-k^2}\text{d}k \leq 3\sqrt{\pi},
\end{align}
which is some constant.

\subsection{Proof of Theorem \ref{thm:graph_geneneralization_gap}}
\label{app:thm1}
Suppose $\bbH_E \in \arg\min_{\bbH\in\ccalH} R_E(\bbH)$, we have 

\begin{align}
    \min_{\bbH\in \ccalH} R_S(\bbH)\leq R_S(\bbH_E), \quad \min_{\bbH\in \ccalH} R_E(\bbH) = R_E(\bbH_E),
\end{align}
which leads to $GA= \min_{\bbH\in \ccalH} R_S(\bbH) - R_E(\bbH_E)$ bounded as
\begin{align}
    &\GeneralizationGap \leq R_S(\bbH_E) - R_E(\bbH_E)\\
    & = \mathbb{E}_{X_N}\left[ \ell\left( \bm\Phi(\bbH_E,\bbL_N, \bbx_{N}), \bby_{N}\right)\right] -  \ell(\bm\Phi(\bbH_E,\bbL_N,\bbx_{N}),\bby_{N}).
\end{align}
We can now subtract and add an intermediate-term  $\mathbb{E}[ \ell(\bbP_N \bm\Phi(\bbH_E, \ccalL,f),\bbP_N g)]$, where the expectation is taken over all the random sampling operator $\bbP_N$ that results in the random sampling point set $X_N$. The $GA$ term therefore can be bounded by 
\begin{align}
&\nonumber \GeneralizationGap \nonumber\leq\Bigg(\mathbb{E}_{X_N}\left[ \ell\left(  \bm\Phi(\bbH_E,\bbL_N,\bbx_{N}), \bby_{N}\right)\right] \\
& \nonumber \qquad \qquad \qquad \qquad \qquad - \mathbb{E}[\ell(\bbP_N \bm\Phi(\bbH_E, \ccalL,f),\bbP_N g)]\Bigg) +  \\
    &   \Bigg( \mathbb{E}[\ell(\bbP_N \bm\Phi(\bbH_E, \ccalL,f),\bbP_N g)]  - \ell(\bm\Phi(\bbH_E,\bbL_N,\bbx_{N}),\bby_{N})\Bigg). 
\end{align}
By absolute value inequality, we have
\begin{align}
&\nonumber \GeneralizationGap \nonumber\leq\Bigg|\mathbb{E}_{X_N}\left[ \ell\left(  \bm\Phi(\bbH_E,\bbL_N,\bbx_{N}), \bby_{N}\right)\right] \\
& \nonumber \qquad \qquad \qquad \qquad \qquad - \mathbb{E}[\ell(\bbP_N \bm\Phi(\bbH_E, \ccalL,f),\bbP_N g)]\Bigg| +  \\
    &   \left| \mathbb{E}[\ell(\bbP_N \bm\Phi(\bbH_E, \ccalL,f),\bbP_N g)]  - \ell(\bm\Phi(\bbH_E,\bbL_N,\bbx_{N}),\bby_{N})\right|.  \label{eqn:GA-decompose}
\end{align}
The \textbf{first term} in \eqref{eqn:GA-decompose} can be analyzed by combining the expectation as the random sampling operator leads to the random sampling set.
\begin{align}
   \nonumber  &\left|\mathbb{E}_{X_N}\left[ \ell\left(  \bm\Phi(\bbH_E,\bbL_N,\bbx_{N}), \bby_{N}\right)\right] - \mathbb{E}[\ell(\bbP_N \bm\Phi(\bbH_E, \ccalL,f),\bby_N)]\right| \\
  & = \left| \mathbb{E}[\ell\left(  \bm\Phi(\bbH_E,\bbL_N,\bbx_{N}), \bby_{N}\right) - \ell(\bbP_N \bm\Phi(\bbH_E, \ccalL,f),\bby_N)] \right|  \\
  & \leq \mathbb{E}\left[ \left| \ell\left(  \bm\Phi(\bbH_E,\bbL_N,\bbx_{N}), \bby_{N}\right) - \ell(\bbP_N \bm\Phi(\bbH_E, \ccalL,f),\bby_N)\right|\right]
\end{align}
Considering the Lipschitz continuity assumption of the loss function in Assumption \ref{ass:loss}, we have
\begin{align}
   & \mathbb{E}\left[ \left| \ell\left(  \bm\Phi(\bbH_E,\bbL_N,\bbx_{N}), \bby_{N}\right) - \ell(\bbP_N \bm\Phi(\bbH_E, \ccalL,f),\bby_N)\right|\right]\\
    &\leq \mathbb{E}\left[ \|\bm\Phi(\bbH_E,\bbL_N,\bbx_{N}) - \bbP_N \bm\Phi(\bbH_E, \ccalL,f),\bby_N) \|_2\right],
\end{align}
which can be bounded by the conclusion in Corollary \ref{cor:expect-diff}.

For the \textbf{second term} in \eqref{eqn:GA-decompose}, we subtract and add an intermediate-term $\ell(\bbP_N \bm\Phi(\bbH_E, \ccalL, f), \bbP_N g)$, with the triangle inequality, we have
\begin{align}
    \nonumber & \left| \mathbb{E}[\ell(\bbP_N \bm\Phi(\bbH_E, \ccalL,f),\bbP_N g)]  - \ell(\bm\Phi(\bbH_E,\bbL_N,\bbx_{N}),\bby_{N})\right|\\
    & \nonumber \leq \left| \mathbb{E}[\ell(\bbP_N \bm\Phi(\bbH_E, \ccalL,f),\bbP_N g)] - \ell(\bbP_N \bm\Phi(\bbH_E, \ccalL, f), \bbP_N g) \right|\\
    &+ \left|\ell(\bbP_N \bm\Phi(\bbH_E, \ccalL, f), \bbP_N g)  -\ell(\bm\Phi(\bbH_E,\bbL_N,\bbx_{N}),\bby_{N}) \right|. \label{eqn:decompose-2}
\end{align}
The first term in \eqref{eqn:decompose-2} is the deviation of the random term $\ell(\bbP_N \bm\Phi(\bbH_E, \ccalL, f), \bbP_N g)$, which can be bounded by Chebyshev's inequality as both the output value and the target functions are bounded. The second term in \eqref{eqn:decompose-2} can be bounded based on Assumption \ref{ass:loss} and Proposition \ref{prop:prob-diff} as
\begin{align}
   \nonumber & \left|\ell(\bbP_N \bm\Phi(\bbH_E, \ccalL, f), \bbP_N g)  -\ell(\bm\Phi(\bbH_E,\bbL_N,\bbx_{N}),\bby_{N}) \right|\\
   & \leq \|\bbP_N \bm\Phi(\bbH_E, \ccalL, f) - \bm\Phi(\bbH_E,\bbL_N,\bbx_{N})\|_2
\end{align}
Combining the upper bounds of these \textbf{two terms} and taking the leading orders from \eqref{eqn:6-4} and \eqref{eqn:8-3}, we can achieve the conclusion that 
\begin{align}
    GA = \ccalO\left( \left(\frac{\log \frac{N}{\delta}}{N}\right)^{\frac{1}{d+4}} 
   \right).
\end{align}

%% file: Generalization-TSP-revised.bbl
\begin{thebibliography}{10}
\providecommand{\url}[1]{#1}
\csname url@samestyle\endcsname
\providecommand{\newblock}{\relax}
\providecommand{\bibinfo}[2]{#2}
\providecommand{\BIBentrySTDinterwordspacing}{\spaceskip=0pt\relax}
\providecommand{\BIBentryALTinterwordstretchfactor}{4}
\providecommand{\BIBentryALTinterwordspacing}{\spaceskip=\fontdimen2\font plus
\BIBentryALTinterwordstretchfactor\fontdimen3\font minus
  \fontdimen4\font\relax}
\providecommand{\BIBforeignlanguage}[2]{{%
\expandafter\ifx\csname l@#1\endcsname\relax
\typeout{** WARNING: IEEEtran.bst: No hyphenation pattern has been}%
\typeout{** loaded for the language `#1'. Using the pattern for}%
\typeout{** the default language instead.}%
\else
\language=\csname l@#1\endcsname
\fi
#2}}
\providecommand{\BIBdecl}{\relax}
\BIBdecl

\bibitem{wangcervino2024}
Z.~Wang, J.~Cerviño, and A.~Ribeiro, ``Generalization of geometric graph
  neural networks,'' \emph{accepted at Asilomar Conference on Signals, Systems,
  and Computers 2024.}

\bibitem{sandryhaila2013discrete}
A.~Sandryhaila and J.~M. Moura, ``Discrete signal processing on graphs,''
  \emph{IEEE transactions on signal processing}, vol.~61, no.~7, pp.
  1644--1656, 2013.

\bibitem{gama2019convolutional}
F.~Gama, A.~G. Marques, G.~Leus, and A.~Ribeiro, ``Convolutional neural network
  architectures for signals supported on graphs,'' \emph{IEEE Transactions on
  Signal Processing}, vol.~67, no.~4, pp. 1034--1049, 2019.

\bibitem{ortega2018graph}
A.~Ortega, P.~Frossard, J.~Kova{\v{c}}evi{\'c}, J.~M. Moura, and
  P.~Vandergheynst, ``Graph signal processing: Overview, challenges, and
  applications,'' \emph{Proceedings of the IEEE}, vol. 106, no.~5, pp.
  808--828, 2018.

\bibitem{wu2022graph}
S.~Wu, F.~Sun, W.~Zhang, X.~Xie, and B.~Cui, ``Graph neural networks in
  recommender systems: a survey,'' \emph{ACM Computing Surveys}, vol.~55,
  no.~5, pp. 1--37, 2022.

\bibitem{yin2023coco}
N.~Yin, L.~Shen, M.~Wang, L.~Lan, Z.~Ma, C.~Chen, X.-S. Hua, and X.~Luo,
  ``Coco: A coupled contrastive framework for unsupervised domain adaptive
  graph classification,'' in \emph{International Conference on Machine
  Learning}.\hskip 1em plus 0.5em minus 0.4em\relax PMLR, 2023, pp.
  40\,040--40\,053.

\bibitem{gosrich2022coverage}
W.~Gosrich, S.~Mayya, R.~Li, J.~Paulos, M.~Yim, A.~Ribeiro, and V.~Kumar,
  ``Coverage control in multi-robot systems via graph neural networks,'' in
  \emph{2022 International Conference on Robotics and Automation (ICRA)}.\hskip
  1em plus 0.5em minus 0.4em\relax IEEE, 2022, pp. 8787--8793.

\bibitem{bronstein2017geometric}
M.~M. Bronstein, J.~Bruna, Y.~LeCun, A.~Szlam, and P.~Vandergheynst,
  ``Geometric deep learning: going beyond euclidean data,'' \emph{IEEE Signal
  Processing Magazine}, vol.~34, no.~4, pp. 18--42, 2017.

\bibitem{sharma2022scaling}
U.~Sharma and J.~Kaplan, ``Scaling laws from the data manifold dimension,''
  \emph{Journal of Machine Learning Research}, vol.~23, no.~9, pp. 1--34, 2022.

\bibitem{CERVINO_ICML}
\BIBentryALTinterwordspacing
J.~Cervino, L.~F.~O. Chamon, B.~D. Haeffele, R.~Vidal, and A.~Ribeiro,
  ``Learning globally smooth functions on manifolds,'' in \emph{Proceedings of
  the 40th International Conference on Machine Learning}, ser. Proceedings of
  Machine Learning Research, A.~Krause, E.~Brunskill, K.~Cho, B.~Engelhardt,
  S.~Sabato, and J.~Scarlett, Eds., vol. 202.\hskip 1em plus 0.5em minus
  0.4em\relax PMLR, 23--29 Jul 2023, pp. 3815--3854. [Online]. Available:
  \url{https://proceedings.mlr.press/v202/cervino23a.html}
\BIBentrySTDinterwordspacing

\bibitem{wang2022convolutional}
Z.~Wang, L.~Ruiz, and A.~Ribeiro, ``Convolutional neural networks on manifolds:
  From graphs and back,'' in \emph{2022 56th Asilomar Conference on Signals,
  Systems, and Computers}.\hskip 1em plus 0.5em minus 0.4em\relax IEEE, 2022,
  pp. 356--360.

\bibitem{wang2023geometric}
------, ``Geometric graph filters and neural networks: Limit properties and
  discriminability trade-offs,'' \emph{IEEE Transactions on Signal Processing},
  2024.

\bibitem{wang2024stability}
------, ``Stability to deformations of manifold filters and manifold neural
  networks,'' \emph{IEEE Transactions on Signal Processing}, pp. 1--15, 2024.

\bibitem{ruiz2020graphon}
L.~Ruiz, L.~Chamon, and A.~Ribeiro, ``Graphon neural networks and the
  transferability of graph neural networks,'' \emph{Advances in Neural
  Information Processing Systems}, vol.~33, pp. 1702--1712, 2020.

\bibitem{ruiz2023transferability}
L.~Ruiz, L.~F. Chamon, and A.~Ribeiro, ``Transferability properties of graph
  neural networks,'' \emph{IEEE Transactions on Signal Processing}, 2023.

\bibitem{levie2021transferability}
R.~Levie, W.~Huang, L.~Bucci, M.~Bronstein, and G.~Kutyniok, ``Transferability
  of spectral graph convolutional neural networks,'' \emph{Journal of Machine
  Learning Research}, vol.~22, no. 272, pp. 1--59, 2021.

\bibitem{maskey2023transferability}
S.~Maskey, R.~Levie, and G.~Kutyniok, ``Transferability of graph neural
  networks: an extended graphon approach,'' \emph{Applied and Computational
  Harmonic Analysis}, vol.~63, pp. 48--83, 2023.

\bibitem{CERVINO_ICASSP}
J.~Cerviño, L.~Ruiz, and A.~Ribeiro, ``Learning by transference: Training
  graph neural networks on growing graphs,'' \emph{IEEE Transactions on Signal
  Processing}, vol.~71, pp. 233--247, 2023.

\bibitem{CERVINO_TSP}
------, ``Training graph neural networks on growing stochastic graphs,'' in
  \emph{ICASSP 2023 - 2023 IEEE International Conference on Acoustics, Speech
  and Signal Processing (ICASSP)}, 2023, pp. 1--5.

\bibitem{scarselli2018vapnik}
F.~Scarselli, A.~C. Tsoi, and M.~Hagenbuchner, ``The vapnik--chervonenkis
  dimension of graph and recursive neural networks,'' \emph{Neural Networks},
  vol. 108, pp. 248--259, 2018.

\bibitem{verma2019stability}
S.~Verma and Z.-L. Zhang, ``Stability and generalization of graph convolutional
  neural networks,'' in \emph{Proceedings of the 25th ACM SIGKDD International
  Conference on Knowledge Discovery \& Data Mining}, 2019, pp. 1539--1548.

\bibitem{zhou2021generalization}
X.~Zhou and H.~Wang, ``The generalization error of graph convolutional networks
  may enlarge with more layers,'' \emph{Neurocomputing}, vol. 424, pp. 97--106,
  2021.

\bibitem{yang2024deeper}
G.~Yang, M.~Li, H.~Feng, and X.~Zhuang, ``Deeper insights into deep graph
  convolutional networks: Stability and generalization,'' \emph{arXiv preprint
  arXiv:2410.08473}, 2024.

\bibitem{yang2024bridging}
G.~Yang, J.~Li, M.~Li, H.~Feng, and D.-X. Zhou, ``Bridging smoothness and
  approximation: Theoretical insights into over-smoothing in graph neural
  networks,'' \emph{arXiv preprint arXiv:2407.01281}, 2024.

\bibitem{ma2021subgroup}
J.~Ma, J.~Deng, and Q.~Mei, ``Subgroup generalization and fairness of graph
  neural networks,'' \emph{Advances in Neural Information Processing Systems},
  vol.~34, pp. 1048--1061, 2021.

\bibitem{esser2021learning}
P.~Esser, L.~Chennuru~Vankadara, and D.~Ghoshdastidar, ``Learning theory can
  (sometimes) explain generalisation in graph neural networks,'' \emph{Advances
  in Neural Information Processing Systems}, vol.~34, pp. 27\,043--27\,056,
  2021.

\bibitem{cong2021provable}
W.~Cong, M.~Ramezani, and M.~Mahdavi, ``On provable benefits of depth in
  training graph convolutional networks,'' \emph{Advances in Neural Information
  Processing Systems}, vol.~34, pp. 9936--9949, 2021.

\bibitem{tang2023towards}
H.~Tang and Y.~Liu, ``Towards understanding generalization of graph neural
  networks,'' in \emph{International Conference on Machine Learning}.\hskip 1em
  plus 0.5em minus 0.4em\relax PMLR, 2023, pp. 33\,674--33\,719.

\bibitem{yehudai2021local}
G.~Yehudai, E.~Fetaya, E.~Meirom, G.~Chechik, and H.~Maron, ``From local
  structures to size generalization in graph neural networks,'' in
  \emph{International Conference on Machine Learning}.\hskip 1em plus 0.5em
  minus 0.4em\relax PMLR, 2021, pp. 11\,975--11\,986.

\bibitem{shi2024homophily}
C.~Shi, L.~Pan, H.~Hu, and I.~Dokmani{\'c}, ``Homophily modulates double
  descent generalization in graph convolution networks,'' \emph{Proceedings of
  the National Academy of Sciences}, vol. 121, no.~8, p. e2309504121, 2024.

\bibitem{liao2020pac}
R.~Liao, R.~Urtasun, and R.~Zemel, ``A pac-bayesian approach to generalization
  bounds for graph neural networks,'' in \emph{International Conference on
  Learning Representations}, 2020.

\bibitem{garg2020generalization}
V.~Garg, S.~Jegelka, and T.~Jaakkola, ``Generalization and representational
  limits of graph neural networks,'' in \emph{International Conference on
  Machine Learning}.\hskip 1em plus 0.5em minus 0.4em\relax PMLR, 2020, pp.
  3419--3430.

\bibitem{maskey2022generalization}
S.~Maskey, R.~Levie, Y.~Lee, and G.~Kutyniok, ``Generalization analysis of
  message passing neural networks on large random graphs,'' \emph{Advances in
  neural information processing systems}, vol.~35, pp. 4805--4817, 2022.

\bibitem{maskey2024generalization}
S.~Maskey, G.~Kutyniok, and R.~Levie, ``Generalization bounds for message
  passing networks on mixture of graphons,'' \emph{arXiv preprint
  arXiv:2404.03473}, 2024.

\bibitem{talmon2015manifold}
R.~Talmon, S.~Mallat, H.~Zaveri, and R.~R. Coifman, ``Manifold learning for
  latent variable inference in dynamical systems,'' \emph{IEEE Transactions on
  Signal Processing}, vol.~63, no.~15, pp. 3843--3856, 2015.

\bibitem{peyre2009manifold}
G.~Peyr{\'e}, ``Manifold models for signals and images,'' \emph{Computer vision
  and image understanding}, vol. 113, no.~2, pp. 249--260, 2009.

\bibitem{osher2017low}
S.~Osher, Z.~Shi, and W.~Zhu, ``Low dimensional manifold model for image
  processing,'' \emph{SIAM Journal on Imaging Sciences}, vol.~10, no.~4, pp.
  1669--1690, 2017.

\bibitem{wangcervnino2024neurips}
Z.~Wang, J.~Cerviño, and A.~Ribeiro, ``A manifold perspective on the
  statistical generalization of graph neural networks,'' \emph{arXiv preprint
  arXiv:2406.05225}, 2024.

\bibitem{grigor2006heat}
A.~Grigor’yan, ``Heat kernels on weighted manifolds and applications,''
  \emph{Cont. Math}, vol. 398, no. 2006, pp. 93--191, 2006.

\bibitem{calder2022improved}
J.~Calder and N.~G. Trillos, ``Improved spectral convergence rates for graph
  laplacians on $\varepsilon$-graphs and k-nn graphs,'' \emph{Applied and
  Computational Harmonic Analysis}, vol.~60, pp. 123--175, 2022.

\bibitem{ramakrishna2020user}
R.~Ramakrishna, H.-T. Wai, and A.~Scaglione, ``A user guide to low-pass graph
  signal processing and its applications: Tools and applications,'' \emph{IEEE
  Signal Processing Magazine}, vol.~37, no.~6, pp. 74--85, 2020.

\bibitem{arendt2009weyl}
W.~Arendt, R.~Nittka, W.~Peter, and F.~Steiner, ``Weyl’s law: Spectral
  properties of the laplacian in mathematics and physics,'' \emph{Mathematical
  analysis of evolution, information, and complexity}, pp. 1--71, 2009.

\bibitem{van2008visualizing}
L.~Van~der Maaten and G.~Hinton, ``Visualizing data using t-sne.''
  \emph{Journal of machine learning research}, vol.~9, no.~11, 2008.

\bibitem{mikolov2013distributed}
T.~Mikolov, I.~Sutskever, K.~Chen, G.~S. Corrado, and J.~Dean, ``Distributed
  representations of words and phrases and their compositionality,''
  \emph{Advances in neural information processing systems}, vol.~26, 2013.

\bibitem{wang2020microsoft}
K.~Wang, Z.~Shen, C.~Huang, C.-H. Wu, Y.~Dong, and A.~Kanakia, ``Microsoft
  academic graph: When experts are not enough,'' \emph{Quantitative Science
  Studies}, vol.~1, no.~1, pp. 396--413, 2020.

\bibitem{yang2016revisiting}
Z.~Yang, W.~Cohen, and R.~Salakhudinov, ``Revisiting semi-supervised learning
  with graph embeddings,'' in \emph{International conference on machine
  learning}.\hskip 1em plus 0.5em minus 0.4em\relax PMLR, 2016, pp. 40--48.

\bibitem{von2008consistency}
U.~Von~Luxburg, M.~Belkin, and O.~Bousquet, ``Consistency of spectral
  clustering,'' \emph{The Annals of Statistics}, pp. 555--586, 2008.

\end{thebibliography}
